\documentclass[11pt,letterpaper]{article}
\pdfoutput=1

\usepackage[linesnumbered,boxruled,vlined]{algorithm2e}

\usepackage[utf8]{inputenc}
\usepackage[T1]{fontenc}

\usepackage{cite}
\usepackage{graphicx}
\usepackage{threeparttable}
\usepackage{booktabs}
\usepackage[margin=1in]{geometry}
\usepackage{amsmath,setspace,amssymb}
\usepackage{amsthm}
\usepackage{latexsym}
\usepackage{comment}
\usepackage{xspace}
\usepackage{enumitem}
\setlist{itemsep=2pt, topsep=4pt, parsep=0pt}
\setlist[enumerate]{label=(\arabic*)}
\usepackage[dvipsnames]{xcolor}
\usepackage[unicode,pdfencoding=auto,colorlinks=true]{hyperref}
\colorlet{LinkRed}{red!50!black}
\colorlet{CiteGreen}{green!50!black}
\colorlet{UrlBlue}{blue!50!black}
\hypersetup{linkcolor=LinkRed, citecolor=CiteGreen, urlcolor=UrlBlue}
\usepackage{bookmark}
\usepackage{thm-restate}
\usepackage{here}
\usepackage[capitalise,noabbrev,nameinlink]{cleveref}
\usepackage{subcaption}

\graphicspath{{./}}

\theoremstyle{plain}
\newtheorem{theorem}{Theorem}
\newtheorem{lemma}[theorem]{Lemma}
\newtheorem{corollary}[theorem]{Corollary}

\newtheorem{claim}[theorem]{Claim}

\theoremstyle{definition}
\newtheorem{definition}[theorem]{Definition}

\newcommand{\dist}{\mathsf{dist}}
\newcommand{\Ot}{\tilde{O}}
\global\long\def\tO{\tilde{O}}%

\newcommand{\tabnote}[1]{{\fontsize{8pt}{9.6pt}\selectfont #1}}

\DontPrintSemicolon
\SetKwProg{Fn}{Procedure}{}{end}
\SetKwFunction{BlockMetric}{BlockMetric}
\SetKwFunction{BackwardWitnesses}{BackwardWitnesses}
\SetKwFunction{PairKernel}{PairKernel}
\SetKwFunction{MSRP}{MSRP}
\SetKwFunction{CanonicalMSRP}{CanonicalMSRP}
\SetKwFunction{SSRP}{SSRP}
\SetKwFunction{BFS}{BFS}
\SetKwFunction{DetRP}{DeterministicRP}
\crefname{AlgoLine}{line}{lines}
\Crefname{AlgoLine}{Line}{Lines}
\SetKwFunction{VWX}{DualFailureOnPath}
\SetKwInOut{Input}{Input}
\SetKwInOut{Output}{Output}
\SetKw{KwAnd}{and}
\SetKw{KwDownTo}{downto}
\SetKw{KwOr}{or}
\SetKw{KwBreak}{break}

\usepackage{tikz}
\usetikzlibrary{calc,intersections}

\allowdisplaybreaks[2]
\newcommand{\randorder}{\textcircled{r}}
\newcommand{\randordernote}{The author ordering was randomized using \url{https://www.aeaweb.org/journals/policies/random-author-order/generator}.}

\title{Simplified and Faster Algorithms for the Directed Single/Dual-Failure Replacement Paths \\ via Modulo-Layer Hitting Sets}
\author{
Kaito Harada \\
The University of Osaka\\
\texttt{k-harada@ist.osaka-u.ac.jp} \and
\randorder\thanks{\randordernote} \and
Tatsuya Terao \\
Kyoto University\\
\texttt{ttatsuya@kurims.kyoto-u.ac.jp}
}

\date{}

\begin{document}

\maketitle
\thispagestyle{empty}
\setcounter{footnote}{1}

\begin{abstract}
We study the replacement paths problem in unweighted directed graphs with $n$ vertices and $m$ edges under one or two edge failures. We give simpler and faster algorithms for both settings. Both algorithms are based on a new technique, which we call \emph{modulo-layer hitting sets}.

For the single-failure setting, we obtain a simple deterministic combinatorial algorithm running in $O(m\sqrt{n})$ time, which is nearly optimal under the combinatorial Boolean Matrix Multiplication (BMM) conjecture. In addition to removing the logarithmic overhead of the deterministic $\tilde{O}(m\sqrt{n})$-time\footnote{The $\tilde{O}(\cdot)$ notation hides polylogarithmic factors, i.e., $\tilde{O}(f(n)) = O(f(n)\,\mathrm{polylog}(n))$.} algorithm of Alon, Chechik, and Cohen~[ICALP'19], our algorithm admits a much simpler implementation and analysis.

For the dual-failure setting, we obtain a randomized combinatorial algorithm running in $\tilde{O}(n^{3-1/4})$ time. This improves upon the previous $\tilde{O}(n^{3-1/18})$-time combinatorial algorithm of Chechik and Zhang~[ICALP'24].
\end{abstract}

\section{Introduction}

In the replacement paths problem, we study how shortest paths change under edge failure.
Specifically, given an unweighted directed graph $G=(V,E)$ with $n$ vertices and $m$ edges and two vertices $s,t\in V$, the classical replacement paths (RP) problem asks to compute, for every edge $e\in E$, a shortest $s$-$t$ path that avoids $e$. 
Its natural two-failure extension, the dual-failure replacement paths (2-FRP) problem, asks to compute, for every pair of edges $e,e'\in E$, a shortest $s$-$t$ path that avoids both $e$ and $e'$.

The replacement paths problem has been extensively studied in the literature~\cite{malik1989k,nardelli2001faster,hershberger2007difficulty,gotthilf2009improved,emek2010near,wulff2010solving,bernstein2010nearly,erickson2011computing,williams2011faster,RZ12,weimann2013replacement,lee2014replacement,grandoni2019faster,chechik2020simplifying,VWX22,CZ24,manoharan2024computing, chang2025optimal,NV26, HI26}, motivated in part by its role as a basic model for fault-tolerant routing. 
In real-world networks, links may become unavailable due to failures, and it is desirable to quickly find alternative routes while preserving shortest-path guarantees. 
Replacement paths naturally capture this task by computing shortest routes that remain available after such failures.
Beyond this networking motivation, replacement paths are also closely connected to several fundamental problems.
In particular, replacement paths arise in computing Vickrey pricing for shortest paths in algorithmic mechanism design~(see \cite{nisan1999algorithmic, hershberger2001vickrey}).
They are also closely related to the $k$ shortest paths problem, which is itself a classical and well-studied problem~\cite{lawler1972procedure,eppstein1998finding}.

For the classical RP problem, the complexity differs significantly between undirected and directed graphs.
Note that the output size of RP is only $\Theta(n)$, since it suffices to consider failures of edges on a fixed shortest $s$--$t$ path.
For undirected graphs, RP can be solved essentially as efficiently as single-source shortest paths~\cite{malik1989k,lee2014replacement}.
In contrast, RP on directed graphs appears substantially harder.
For unweighted directed graphs, Roditty and Zwick~\cite{RZ12} gave a randomized combinatorial algorithm running in $\tO(m\sqrt n)$ time.
Obtaining a deterministic algorithm with a comparable running time remained open for more than a decade, until Alon, Chechik, and Cohen~\cite{ACC19} gave a deterministic $\tO(m\sqrt n)$-time combinatorial algorithm.
This running time is optimal up to $n^{o(1)}$ factor among combinatorial algorithms, assuming the combinatorial Boolean Matrix Multiplication (BMM) conjecture~\cite{williams2018subcubic}.
For weighted directed graphs, RP is subcubic-equivalent to APSP~\cite{williams2018subcubic}.\footnote{
Very recently, Alman and Vassilevska Williams~\cite{alman2026trulysubquadratic3sumtruly} gave a truly subcubic $O(n^{2.9995})$-time algorithm for APSP on directed graphs with polynomially bounded integer weights, refuting the APSP hypothesis. 
By the subcubic equivalence~\cite{williams2018subcubic}, RP on weighted directed graphs is thus also solvable in truly subcubic time. 
Their algorithm relies on fast matrix multiplication and is not combinatorial.
In particular, the conditional lower bounds in this paper based on the combinatorial BMM conjecture are unaffected.}
For algebraic algorithms, Vassilevska Williams~\cite{williams2011faster} gave a randomized $\tO(Mn^\omega)$-time algorithm\footnote{
Here, $\omega$ denotes the exponent of square matrix multiplication.
The current best bound over an arbitrary fixed field is $\omega<2.37106$~\cite{ADVXXZ25,DEKMRSZAVB26,OAI:Staggered-extraction-for-exact-matrix-multiplication-over-every-field-September-24-2026}.
Over the complex numbers, the stronger bound $\omega(\mathbb{C})\le 9/4$ was recently established~\cite{OAI:Matrix-Multiplication-Nine-Fourths-October-2-2026}.
} for edge weights in $\{-M,\ldots,M\}$.

Although the notion of a ``combinatorial algorithm'' is not formally defined, here and throughout we use the term informally to refer to algorithms that do not rely on fast matrix multiplication, such as Strassen-like methods.
Such algorithms are often viewed as more efficient in practice, since fast matrix multiplication algorithms typically involve large hidden constants.
We refer the reader to~\cite[Section~1.1]{AFKLM24} for a more extensive discussion of the motivations for seeking combinatorial algorithms, as well as for possible operational definitions of this notion.

The 2-FRP problem is considerably less understood.
Vassilevska Williams, Woldeghebriel, and Xu~\cite{VWX22} gave a $\tO(n^3)$-time algorithm for weighted directed graphs and an $n^{8/3-o(1)}$ conditional lower bound for combinatorial algorithms on unweighted directed graphs.
Chechik and Zhang~\cite{CZ24} subsequently gave the first truly subcubic combinatorial algorithm for unweighted directed graphs, with running time $\tO(n^{3-1/18})$.
For algebraic algorithms, they also obtained a $\tO(Mn^{2.8716})$-time algorithm for edge weights in $\{-M,\ldots,M\}$.
Thus, for combinatorial algorithms on unweighted directed graphs, there remains a substantial gap between the known upper and lower bounds.
In contrast, for undirected graphs, Nogler and Vassilevska Williams~\cite{NV26} recently reduced 2-FRP to single source replacement paths, yielding a $\tO(m\sqrt n+n^2)$-time algorithm for unweighted graphs.
They left open whether a similar reduction, possibly in a weaker form, can be obtained for directed graphs.
We note that, although 2-FRP considers all pairs of failed edges, its output size is only $\Theta(n^2)$; see \cref{subsec:replacement-paths-variants} for details.

\subsection{Our Results}
\cref{tab:known-results} summarizes the best known upper and lower bounds for RP and 2-FRP, together with our new results.
Our first result is a simple deterministic algorithm for the RP problem.

\begin{theorem}\label{thm:rp}
    There is a deterministic combinatorial algorithm for the RP problem on unweighted directed graphs that runs in $O(m\sqrt{n})$ time.
\end{theorem}
The running time in \cref{thm:rp} is optimal up to an $n^{o(1)}$ factor among combinatorial algorithms~\cite{williams2018subcubic}, and removes the logarithmic overhead in the previous randomized and deterministic algorithms~\cite{RZ12,ACC19}.

More importantly, our derandomization is considerably simpler.
The randomized algorithm of Roditty and Zwick~\cite{RZ12} is conceptually simple, whereas its derandomization by Alon, Chechik, and Cohen~\cite{ACC19} requires substantially more machinery, including the construction of critical paths using decremental shortest-path techniques and a greedy hitting-set procedure. 
To obtain a simpler deterministic algorithm, we introduce a new technique, which we call a \emph{modulo-layer hitting set}.
The idea is extremely simple.
After performing BFS from $s$, we partition the vertices into $L = \Theta(\sqrt{n})$ classes according to their distances from $s$ modulo $L$, and use a smallest class, which has size at most $n/L = O(\sqrt{n})$, as the hitting set.
We believe that the resulting algorithm is simple enough for textbook presentation and practical implementation.

Combined with the reduction of Roditty and Zwick~\cite{RZ12}, \cref{thm:rp} also yields an efficient algorithm for computing the $k$ shortest simple $s$--$t$ paths.

\begin{corollary}\label{cor:kssp}
    There is a deterministic combinatorial algorithm for the $k$ simple shortest paths problem on unweighted directed graphs that runs in $O(km\sqrt{n})$ time.
\end{corollary}
Roditty and Zwick~\cite{RZ12} reduced the $k$ simple shortest paths problem to $O(k)$ instances of the replacement paths problem.
Thus, \cref{cor:kssp} follows from \cref{thm:rp} without any logarithmic overhead.

Our second result is a faster combinatorial algorithm for the 2-FRP problem.
Our algorithm builds on our new modulo-layer hitting-set technique introduced for RP and combines it with additional ideas.

\begin{theorem}\label{thm:2frp}
    There is a randomized combinatorial algorithm for the 2-FRP problem on unweighted directed graphs that runs in $\Ot(n^{3 - 1/4})$ time.
\end{theorem}
Our algorithm improves the previous combinatorial bound $\tilde{O}(n^{3-1/18})$ and is simpler than the previous algorithm of Chechik and Zhang~\cite{CZ24}.
Our algorithm also reduces directed 2-FRP to a collection of multiple-source replacement paths (MSRP) instances.
This reduction can be viewed as a partial answer to the question of Nogler and Vassilevska Williams~\cite{NV26}.
We then solve the resulting MSRP instances using the recent algorithm of Harada and Izumi~\cite{HI26} for directed graphs.

Our results demonstrate that our new modulo-layer hitting-set technique is useful for both directed RP and 2-FRP. 
We believe that this simple technique may also find applications to other related shortest-path problems.

\begin{table}[t]
    \centering
    \caption{Best known results and our results for RP and 2-FRP.
    Lower bounds marked with $\dagger$ apply to combinatorial algorithms and are conditional on the combinatorial BMM conjecture.}\label{tab:known-results}
    \setlength{\tabcolsep}{4pt}
    \begin{tabular}{ccccc}
        \toprule
        \textbf{Direction} & \textbf{Weights} & \textbf{Det./Rand.} & \textbf{Upper Bound} & \textbf{Lower Bound} \\
        \midrule
        \multicolumn{5}{l}{\textit{Replacement Paths (RP)}} \\
        \addlinespace[2pt]
        undir. & $[1, n^c]$ & det. & $O(m + n\log n)$~\tabnote{\cite{malik1989k}} & $\Omega(m+n)$ \\
        dir. & $[1, n^c]$ & det. & APSP-time~\tabnote{\cite{williams2018subcubic}} & $\Omega(m+n)$ \\
        dir. & $[-M, M]$ & rand. & $\tilde{O}(Mn^{\omega})$~\tabnote{\cite{williams2011faster}} & $\Omega(m+n)$ \\
        undir. & unweigh. & det. & $O(m + n)$~\tabnote{\cite{lee2014replacement}} & $\Omega(m+n)$ \\
        dir. & unweigh. & det. & $O(m\sqrt{n})$~\tabnote{[\cref{thm:rp}]} & $mn^{1/2 - o(1)}{}^{\dagger}$~\tabnote{\cite{williams2018subcubic}} \\
        \midrule
        \multicolumn{5}{l}{\textit{Dual-Failure Replacement Paths (2-FRP)}} \\
        \addlinespace[2pt]
        undir. & $[1, n^c]$ & det. & APSP-time~\tabnote{\cite{NV26}} & APSP-time~\tabnote{\cite{henzinger2017hardness,shucheng20252_3frp}} \\
        dir. & $[1, n^c]$ & det. & $\tilde{O}(n^3)$~\tabnote{\cite{VWX22}} & APSP-time~\tabnote{\cite{williams2018subcubic}} \\
        undir. & $[1, M]$ & rand. & $\tilde{O}(Mn^{\omega})$~\tabnote{\cite{NV26}} & $\Omega(n^2)$ \\
        dir. & $[-M, M]$ & rand. & $\tilde{O}(Mn^{2.8716})$~\tabnote{\cite{CZ24}} & $\Omega(n^2)$ \\
        undir. & unweigh. & det. & $\tilde{O}(m\sqrt{n} +n^2)$~\tabnote{\cite{NV26}} & $mn^{1/2-o(1)}{}^{\dagger}$~\tabnote{\cite{NV26}} \\
        dir. & unweigh. & rand. & $\tilde{O}(n^{3-1/4})$~\tabnote{[\cref{thm:2frp}]} & $n^{8/3 - o(1)}{}^{\dagger}$~\tabnote{\cite{VWX22}} \\
        \bottomrule
    \end{tabular}
\end{table}

\subsection{Additional Related Work} \label{subsec:related_work}

Fast algorithms for replacement paths have also been studied on special graph classes, including planar graphs~\cite{emek2010near,wulff2010solving} and surface embedded graphs~\cite{erickson2011computing}.
Another line of work considers fast approximation algorithms~\cite{bernstein2010nearly, chechik2024nearly}.
The replacement paths problem has been studied in distributed settings~\cite{manoharan2024computing, chang2025optimal}.

A closely related line of work studies \emph{distance sensitivity oracles} (DSOs), which can be viewed as a data-structure version of replacement paths.
After preprocessing a graph, a DSO answers shortest-path distance queries efficiently under edge or vertex failures. 
DSOs have also been extensively studied in the literature~\cite{bernstein2009nearly,grandoni2019faster, chechik2020distance,gu_et_al:LIPIcs.ICALP.2021.76,ren2022improved,duan2022maintaining,karczmarz2023sensitivity,dey_et_al:LIPIcs.ESA.2024.45,harada_et_al:LIPIcs.ESA.2024.65}.

\section{Preliminaries}

\subsection{Notation}
Let $G=(V,E)$ be a directed graph with $n=|V|$ vertices and $m=|E|$ edges.
For an edge set $F\subseteq E$, we write $G\setminus F$ for the graph obtained from $G$ by removing all edges in $F$.
For an edge $e \in E$, we simply write $G\setminus e$ instead of $G\setminus\{e\}$.
For a directed graph $D$, let $D^{\mathrm{rev}}$ denote the graph obtained by reversing every edge of $D$.
For a vertex subset $V'\subseteq V(D)$, let $D[V']$ denote the subgraph of $D$ induced by $V'$.

For a directed graph $D$, possibly with edge weights, and vertices $u,v$, we denote by $\dist_D(u,v)$ the length of a shortest $u$--$v$ path in $D$, or its minimum total weight when $D$ is weighted.
We set $\dist_D(u, v) = \infty$ if $u \notin V(D), v \notin V(D)$, or $D$ contains no $u$--$v$ path.

For a directed path or walk $\rho$ in  $G$, we denote its vertex set and edge set by $V(\rho)$ and $E(\rho)$, respectively.
We let $|\rho|$ denote the number of edges in $\rho$.
For any $u, v \in V(\rho)$ where $u$ comes before $v$ on $\rho$, we write $\rho[u,v]$ for the subpath of $\rho$ from $u$ to $v$, and $\dist_\rho(u,v)$ for the length of this subpath.
For two directed paths or walks $\alpha$ and $\beta$ such that the last vertex of $\alpha$ is the first vertex of $\beta$, we denote their concatenation by $\alpha\circ\beta$.

Throughout the paper, let $\pi=\langle v_0,v_1,\ldots,v_k \rangle$ be a fixed shortest path in $G$ from $s = v_0$ to $t = v_{k}$, and let $H:=G\setminus E(\pi)$.
For a path $\tau$, let $H+\tau$ denote the graph obtained from $H$ by adding all edges of $\tau$.

For integers $a, b$ and a positive integer $L$, we write $a \equiv b \pmod{L}$
if $L$ divides $a - b$.

\subsection{Replacement Paths and Variants} \label{subsec:replacement-paths-variants}

\paragraph*{(Single-Failure) Replacement Paths (RP).}
Given a directed graph $G$, two vertices $s,t$, and a shortest $s$--$t$ path $\pi$, the \emph{replacement paths} (RP) problem asks to compute $\dist_{G\setminus e}(s,t)$ for every edge $e\in E$.
It suffices to consider edges on $\pi$, since removing an edge outside $\pi$ does not change the $s$--$t$ distance.
Since $\pi$ contains at most $n-1$ edges, the output size of RP is $O(n)$.

\paragraph*{Dual-Failure Replacement Paths (2-FRP).}
The \emph{dual-failure replacement paths} (2-FRP) problem extends RP to two edge failures. 
The goal is to compute $\dist_{G\setminus\{e,e'\}}(s,t)$ for every pair of edges $e,e'\in E$.

Although there are $O(m^2)$ possible pairs of failed edges, the output size of 2-FRP is only $O(n^2)$. 
Indeed, it suffices to consider $e\in E(\pi)$ and, for each such edge $e$, only the edges $e'$ on a shortest $s$--$t$ path in $G\setminus e$.
All other distances are equal to either $\dist_G(s,t)$ or
$\dist_{G\setminus e}(s,t)$.
Thus, only $O(n^2)$ nontrivial output values need to be represented.

\subsection{Multiple Source Replacement Paths}\label{subsec:msrp}

In the \emph{multiple source replacement paths} (MSRP) problem, we are given a graph and a set $S$ of $\sigma$ sources.
For every source $u \in S$, a shortest-path tree $T_u$ rooted at $u$ is fixed, and the output consists of the distances $\dist_{G \setminus e}(u, v)$ for all $v$ and all edges $e$ on the tree path $T_u[u, v]$.
For $e \notin E(T_u[u, v])$ the tree path survives, so $\dist_{G \setminus e}(u, v) = \dist_G(u, v)$.
The case $\sigma = 1$ is the \emph{single source replacement paths} SSRP problem.

\begin{lemma}[Theorem 1 of \cite{HI26}]\label{thm:msrp}
    MSRP on unweighted directed graphs can be solved by a randomized combinatorial algorithm in $\Ot(m\sqrt{\sigma n} + \sigma n^2)$ time, with all output distances correct with high probability.
    In particular, setting $\sigma=1$ gives an $\Ot(m\sqrt{n}+n^2)$-time algorithm for SSRP.
\end{lemma}

\subsection{2-FRP for Weighted Directed Graphs}
Our 2-FRP algorithm reduces groups of failure pairs to small weighted instances, on which we run the algorithm of Vassilevska Williams, Woldeghebriel, and Xu~\cite{VWX22} as a black box.
They proposed algorithms for 2-FRP on an $n$-vertex directed graph with integer edge weights and no negative cycles in $O(n^3 \log^2 n)$ time deterministically, or in $O(n^3 \log n)$ time with high probability.

\begin{lemma}[Lemmas 4.3 and 4.4 of~\cite{VWX22}]\label{thm:weighted-2frp}
    Let $G$ be a directed graph with $n$ vertices, polynomially bounded integer edge weights, and no negative cycles, let $s, t \in V(G)$, and let $\pi$ be a shortest $s$--$t$ path in $G$ given as part of the input.
    The distances $\dist_{G \setminus \{e, e'\}}(s, t)$ for all pairs of edges $e, e' \in E(\pi)$ can be computed by a deterministic combinatorial algorithm in $\Ot(n^3)$ time.
\end{lemma}

\subsection{Canonical Replacement Paths}

We use the notion of a \emph{canonical replacement path} introduced in \cite{VWX22}; see also \cite[Definition~3 in the conference version]{CZ24}. 

\begin{definition}[Canonical path]\label{def:canonical}
    Let $\pi$ be a shortest $s$--$t$ path, and let $F \subseteq E$ be a set of failed edges.
    An $s$--$t$ path $\rho$ in $G \setminus F$ is \emph{canonical with respect to $\pi$ and $F$} if, for any $u,v \in V(\pi) \cap V(\rho)$ such that $u$ appears before $v$ in both $\pi, \rho$ and $E(\pi[u,v]) \cap F = \emptyset$, then $\rho[u,v]$ is the same as $\pi[u,v]$. 
\end{definition}

Throughout the paper, in both the single-failure and dual-failure settings, whenever we consider a shortest replacement path, we choose it to be canonical.
Such a canonical shortest replacement path always exists. 
Indeed, among all shortest $s$--$t$ paths in $G\setminus F$, choose $\rho$ minimizing $|E(\rho)\setminus E(\pi)|$. 
If $\rho$ is not canonical, then for some violating pair $u,v$, replacing $\rho[u,v]$ by $\pi[u,v]$ yields another shortest $s$--$t$ path with fewer edges outside $\pi$, which contradicts the choice of $\rho$.

\subsection{Basic Fact on Single-Failure Replacement Paths} \label{subsec:replacement_path}

For two vertices $v_a,v_b\in V(\pi)$ with $a<b$, we call a shortest $v_a$--$v_b$ path in $H$ a \emph{detour}.
A standard structural fact about single-failure replacement paths is that a shortest replacement path can be decomposed into a prefix of $\pi$, followed by a single detour, and then a suffix of $\pi$; see, e.g.,~\cite[Lemma~5]{chechik2020simplifying}.
See \cref{fig:rp-detour} for an illustration.
This gives the following identity, which forms the basis of fast algorithms for single-failure replacement paths.

\begin{lemma} \label{lem:rp-detour}
    Let $H := G \setminus E(\pi)$.
    For every edge $e = (v_i, v_{i+1}) \in E(\pi)$, 
    \[ \dist_{G \setminus e}(s, t) = \min_{0 \le a \le i < b \le k} \bigl\{ \dist_\pi(s, v_a) + \dist_{H}(v_a, v_b) + \dist_\pi(v_b, t) \bigr\}. \]
\end{lemma}

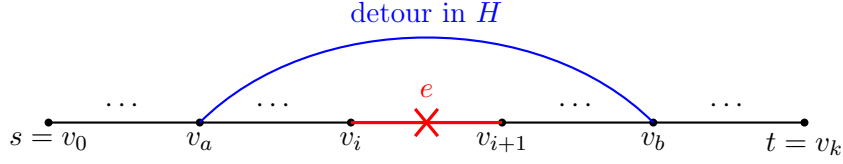
\begin{figure}[t]
    \centering
    \begin{tikzpicture}[x=1cm,y=1cm]

        \coordinate (s)   at (0,0);
        \coordinate (va)  at (2,0);
        \coordinate (vi)  at (4,0);
        \coordinate (vip) at (6,0);
        \coordinate (vb)  at (8,0);
        \coordinate (t)   at (10,0);

        \draw[thick] (s) -- (va) -- (vi) -- (vip) -- (vb) -- (t);

        \fill (s)   circle (1.5pt);
        \fill (va)  circle (1.5pt);
        \fill (vi)  circle (1.5pt);
        \fill (vip) circle (1.5pt);
        \fill (vb)  circle (1.5pt);
        \fill (t)   circle (1.5pt);

        \node[below] at (s)   {$s=v_0$};
        \node[below] at (va)  {$v_a$};
        \node[below] at (vi)  {$v_i$};
        \node[below] at (vip) {$v_{i+1}$};
        \node[below] at (vb)  {$v_b$};
        \node[below] at (t)   {$t=v_k$};

        \node[above] at (1,0) {$\cdots$};
        \node[above] at (3,0) {$\cdots$};
        \node[above] at (7,0) {$\cdots$};
        \node[above] at (9,0) {$\cdots$};

        \draw[red,very thick] (4,0) -- (6,0);
        \draw[red,very thick] (4.85,-0.18) -- (5.15,0.18);
        \draw[red,very thick] (4.85,0.18) -- (5.15,-0.18);
        \node[red,above] at (5,0.2) {$e$};

        \draw[blue,thick]
            (va) .. controls (3.5,1.5) and (6.5,1.5) .. (vb);
        \node[blue,above] at (5,1.2) {detour in $H$};

    \end{tikzpicture}
    \caption{A shortest replacement path avoiding the failed edge
    $e=(v_i,v_{i+1})$.}
    \label{fig:rp-detour}
\end{figure}

\section{Our New Simple Deterministic Algorithm for RP}\label{sec:overview-rp}

We first introduce our new derandomization technique for the RP problem, which we call a \emph{modulo-layer hitting set}.
The same technique will later be applied to the 2-FRP problem.

By \cref{lem:rp-detour}, a central task in designing fast RP algorithms is to efficiently compute the detour distances $\dist_H(v_a,v_b)$ for pairs of vertices $v_a,v_b\in V(\pi)$.
Since there can be $\Theta(n^2)$ such pairs, these distances are usually not computed explicitly. Instead, they are computed implicitly.
We first briefly review how previous algorithms handle these detours, and then describe our approach based on the modulo-layer hitting-set technique.

\paragraph*{Previous randomized approach~\cite{RZ12}.}  
The randomized algorithm of Roditty and Zwick \cite{RZ12} is conceptually simple.
It distinguishes replacement paths whose detour has at most $\lceil \sqrt{n} \rceil$ edges from those whose detour has more than $\lceil \sqrt{n} \rceil$ edges.
Short detours are found deterministically. For long detours, a set $\mathcal{R}$ of \emph{pivots} is chosen by including every vertex independently with probability $c \log n / \sqrt{n}$ for a sufficiently large constant $c$, and the distances in $H$ between the vertices of $\pi$ are computed by a BFS in each direction from every pivot.
Since a long detour has more than $\lceil \sqrt{n} \rceil$ vertices, with high probability every long detour contains a pivot.
Consequently, for every long detour from $v_a$ to $v_b$, we have
$
    \dist_H(v_a,v_b)
    =
    \min_{r\in\mathcal{R}}
    \left\{
        \dist_H(v_a,r)+\dist_H(r,v_b)
    \right\}.
$

\paragraph*{Previous deterministic approach~\cite{ACC19}.} 
Derandomizing the pivot-selection step above is substantially more technically challenging.
Alon, Chechik, and Cohen~\cite{ACC19} obtained the first deterministic $O(m\sqrt n\log n)$-time algorithm.
Their proof and algorithm require constructing a carefully chosen family of $O(n)$ critical paths, each path containing $\lceil \sqrt{n} \rceil$ edges.
Computing this family efficiently requires a variant of the decremental single-source shortest-path algorithm of King~\cite{king1999fully}.
They then apply a greedy hitting-set procedure to the family of critical paths, obtaining a set of $O(\sqrt n\log n)$ pivots that hits every critical path. 
These pivots are subsequently used within the framework of Roditty and Zwick~\cite{RZ12}.
The logarithmic factors in the size of the hitting set and in its computation account for the $\Ot(\cdot)$ in the running time.

\paragraph*{Overview of our approach.} 
Our underlying idea is entirely different, and both the algorithm and its analysis are remarkably simple.
The key difference is that we do not classify detours by their length, as in~\cite{RZ12,ACC19}; instead, we classify them by whether they pass through the hitting set or stay within a small range of consecutive BFS layers from $s$, which we call a block (defined below).

Instead of constructing a hitting set for a carefully chosen family of critical paths, we obtain one directly from a single BFS from $s$.
We group the vertices by their distance from $s$ modulo a parameter $L = \Theta(\sqrt{n})$, and let $\mathcal{R}$ consist of the vertices in a smallest group, which has size at most $n/L = O(\sqrt{n})$.
In other words, for some $0 \le \ell^\star < L$, $\mathcal{R}$ is the union of the BFS layers at distances $h_q := \ell^\star + qL$ from $s$ ($q \ge 0$), where the BFS layer at distance $d$ is the set of vertices at distance exactly $d$ from $s$; we call these layers the \emph{separator layers}.
A \emph{block} consists of two consecutive separator layers and all layers between them, i.e., the vertices at distances $h_{q-1}, h_{q-1} + 1, \dots, h_q$ from $s$ for some $q \ge 0$; see \cref{fig:layers-blocks} for an illustration.

Since $\dist_G(s, y) \le \dist_G(s, x) + 1$ for every edge $(x, y)$ of $G$, the following holds for every $q \ge 0$: any path in a subgraph of $G$ that starts at a vertex at distance less than $h_q$ from $s$ and ends at a vertex at distance at least $h_q$ must contain a vertex at distance exactly $h_q$, i.e., a vertex of the separator layer at distance $h_q$, which belongs to $\mathcal{R}$.
In fact, the same structural property is what lets the technique carry over to the 2-FRP problem, where the graph in which detours are sought changes with the second failure (\cref{sec:overview-2frp}).

Using this structural property, we show that every detour either visits a separator layer, i.e., contains a vertex of $\mathcal{R}$, or stays inside a single block.
Thus, $\mathcal{R}$ serves as a hitting set for all detours that leave their block.
Accordingly, the detour distances can be evaluated using two kinds of BFSs: detours of the former kind are handled by BFSs from and to every vertex of $\mathcal{R}$ in $H$, which take $O((m+n)n/L)$ time in total, and detours of the latter kind by BFSs inside each block from the vertices of $\pi$ in it, which take $O((m+n)L)$ time in total.

Since neither the size of this hitting set nor its construction involves a logarithmic factor, our algorithm runs in $O(m\sqrt{n})$ time, improving the $O(m\sqrt{n}\log n)$ bound of~\cite{ACC19}.

\subsection{Modulo-Layer Hitting Set and Modulo-Layer Decomposition} \label{subsec:modulo_layer}
Throughout, we fix a shortest $s$--$t$ path $\pi$ in $G$.
Fix an integer parameter $1 \le L \le n$, whose value will be chosen later; we will eventually set $L := \lceil \sqrt{n}\, \rceil$.
For $i \ge 0$, the \emph{layer} $i$ is the set $\Lambda_i := \{x \in V \mid \dist_G(s, x) = i\}$ of vertices at distance exactly $i$ from $s$.
For $0\le \ell <L$, let
\[
    \mathcal{R}_\ell:=\{x \in V \mid \dist_{G}(s, x) < \infty
             \text{ and }\dist_{G}(s, x) \equiv \ell \pmod L\}.
\]
Choose $\ell^\star\in\arg\min_\ell |\mathcal{R}_\ell|$ and put $\mathcal{R} := \mathcal{R}_{\ell^\star}$.  
Then $|\mathcal{R}| \le n/L$.
We call the layers $h_q := \ell^\star + qL$ ($q \ge 0$), which are exactly the layers containing the vertices of $\mathcal{R}$, the \emph{separator layers}.
The separator layers divide the set of vertices reachable from $s$ into \emph{blocks}: with the convention $h_{-1} := 0$, for $q \ge 0$, \emph{block} $q$ is
$V_q:=\{x\in V \mid h_{q-1} \le \dist_G(s, x) \le h_q\}$, and we let $H_q:=H[V_q]$.
Recall that $H := G \setminus E(\pi)$.
That is, $V_0$ consists of the vertices of layers $0, \dots, \ell^\star$ and, for $q \ge 1$, $V_q$ consists of the vertices of the $L+1$ consecutive layers $h_{q-1}, \dots, h_q$, from one separator layer to the next, inclusive.
Two consecutive blocks $V_q$ and $V_{q+1}$ share exactly the separator layer $h_q$ (i.e., $V_q \cap V_{q+1} = \Lambda_{h_q}$), and any two non-consecutive blocks are disjoint.
Consequently, every vertex lies in at most two blocks (in exactly two if and only if it belongs to $\mathcal{R}$), and every edge of $H$ lies in at most two of the graphs $H_q$.
Finally, for $q \ge 0$, let $\gamma_q$ be the subpath of $\pi$ induced by $V_q$;
that is, $\gamma_q = \pi[h_{q-1}, \min\{h_q, k\}]$.
\cref{alg:rp-blocks} summarizes this construction and \cref{fig:layers-blocks} illustrates it.

Since a vertex of $\pi$ on a separator layer belongs to two blocks, we fix the following convention.
For $0 \le a \le k$, let $q(a)$ denote the unique index $q \ge 0$ with $h_{q-1} \le a < h_q$.
Explicitly, $q(a) = 0$ if $a < \ell^\star$, and $q(a) = \lfloor (a - \ell^\star)/L \rfloor + 1$ otherwise.
Equivalently, for $a < k$, $q(a)$ is the index of the unique block containing the edge $e_a = (v_a, v_{a+1})$ of $\pi$, that is, $e_a \in E(\gamma_{q(a)})$.
We always have $v_a \in V_{q(a)}$.
If $v_a \in \mathcal{R}$, then $v_a$ also lies in $V_{q(a)-1}$, and $q(a)$ selects the block in which $v_a$ is on the lower separator layer $h_{q(a)-1}$.
Note that $a \le a'$ implies $q(a) \le q(a')$.

\begin{algorithm}[!htb]
    \caption{Modulo-Layer Hitting Set and Modulo-Layer Decomposition.}\label{alg:rp-blocks}
    $\ell^\star \gets \arg\min_{0 \le \ell < L} |\{x \in V \mid \dist_G(s, x) \equiv \ell \pmod L\}|$\;
    $\mathcal{R} \gets \{x \in V \mid \dist_G(s, x) \equiv \ell^\star \pmod L\}$\;
    $h_{-1} \gets 0$ and $h_q \gets \ell^\star + qL$ for $q \ge 0$\;
    $\gamma_q \gets \pi[h_{q-1}, \min\{h_q, k\}]$ for every $q \ge 0$\;
    $V_q \gets \{x \in V \mid h_{q-1} \le \dist_G(s, x) \le h_q\}$ and $H_q \gets H[V_q]$ for every $q \ge 0$ with $V_q \ne \emptyset$\;
\end{algorithm}

\begin{figure}[!tb]
    \centering
    \includegraphics[width=\textwidth]{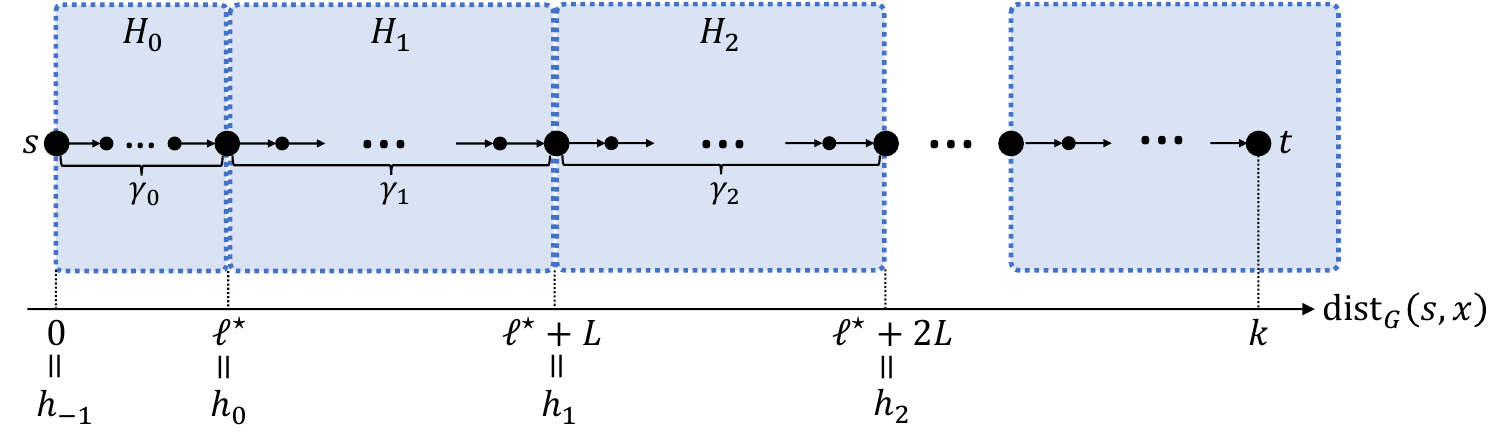}
    \caption{Separator layers and blocks of \cref{alg:rp-blocks}.}\label{fig:layers-blocks}
\end{figure}

Our algorithm computes the detour distances $\dist_H(v_a, v_b)$ using the following lemma.
Recall that $\dist_{H_q}(v_a, v_b) = \infty$ whenever $v_a \notin V_q$ or $v_b \notin V_q$.
\begin{lemma}\label{lem:rp-decomp}
    For every $0 \le a < b \le k$, with $q = q(a)$,
    \[
        \dist_{H}(v_a, v_b) = \min \Bigl\{ \min_{r \in \mathcal{R}} \bigl\{ \dist_H(v_a, r) + \dist_H(r, v_b) \bigr\},\ \dist_{H_q}(v_a, v_b) \Bigr\}.
    \]
\end{lemma}
\begin{proof}
    Every finite term on the right is the length of a walk from $v_a$ to $v_b$ in $H$.
    Hence, the right-hand side is at least $\dist_H(v_a, v_b)$.

    For the converse inequality, we may assume that $\dist_H(v_a, v_b) < \infty$.
    Let $\tau$ be a shortest $v_a$--$v_b$ path in $H$.
    If $\tau$ contains a vertex $r^\ast \in \mathcal{R}$ (possibly $r^\ast = v_a$ or $r^\ast = v_b$), then $\dist_H(v_a, r^\ast) + \dist_H(r^\ast, v_b) = |\tau[v_a, r^\ast]| + |\tau[r^\ast, v_b]| = |\tau| = \dist_H(v_a, v_b)$.
    Otherwise, no vertex of $\tau$ lies on a separator layer, and we claim that every vertex $x$ of $\tau$ satisfies $h_{q-1} < \dist_G(s, x) < h_{q}$.
    Indeed, $\tau$ starts at $v_a$ with $a < h_q$; if some vertex of $\tau$ had layer at least $h_q$, then $\tau$ would contain a vertex of layer exactly $h_q$, which is a separator layer.
    Similarly, $\tau$ ends at $v_b$ with $b > h_{q-1}$; if some vertex of $\tau$ had layer less than $h_{q-1}$, $\tau$ would contain a vertex of layer exactly $h_{q-1}$, which is a separator layer.
    Thus, all vertices of $\tau$ belong to $V_q$, and since $H_q$ is the subgraph of $H$ induced by $V_q$, the path $\tau$ lies in $H_q$.
    Therefore $\dist_{H_q}(v_a, v_b) \le |\tau| = \dist_H(v_a, v_b)$.
\end{proof}

By \cref{lem:rp-decomp}, every detour either visits a separator layer or stays inside one block, so the distance tables of \cref{alg:dist_compute} suffice to evaluate the minimum in \cref{lem:rp-detour}.
Note that the detour distances $\dist_H(v_a, v_b)$ are never computed explicitly: the minimum is evaluated for all failed edges at once by prefix and suffix minima along $\pi$.
The BFSs from and to all pivots $r \in \mathcal{R}$ cost $O((m+n)n/L)$ time since $|\mathcal{R}| \le n/L$, and the BFSs inside the blocks cost $O((m+n)L)$ time since each $\gamma_q$ has at most $L+1$ vertices and every vertex and every edge of $H$ lies in at most two of the graphs $H_q$.
Choosing $L = \lceil \sqrt{n} \rceil$ yields \cref{thm:rp}. 

\begin{algorithm}[!htb]
    \caption{Preprocessing distance computation}\label{alg:dist_compute}
    $\dist_H(r, \ast) \gets \BFS(H, r)$ for each $r \in \mathcal{R}$\;
    $\dist_{H}(\ast, r) \gets \BFS(H^{\mathrm{rev}}, r)$ for each $r \in \mathcal{R}$\;
    $\dist_{H_q}(v, \ast) \gets \BFS(H_q, v)$ for each $0 \le q$ and each $v \in V(\gamma_q)$
\end{algorithm}

\subsection{Overall Algorithm and Its Running Time.}

\paragraph*{Overall algorithm.}
Consider an edge $e_i = (v_i, v_{i+1}) \in E(\pi)$.
Substituting \cref{lem:rp-detour} into \cref{lem:rp-decomp} splits the minimum over $0 \le a \le i < b \le k$ into two candidates.
The first one, the \emph{pivot candidate}, collects the detours that pass through a vertex of $\mathcal{R}$:
\begin{equation*}
    \begin{gathered}
        C_i := \min_{r \in \mathcal{R}} \bigl\{ A_i(r) + B_i(r) \bigr\}, \quad\text{where}\\
        A_i(r) := \min_{0 \le a \le i} \bigl\{ a + \dist_H(v_a, r) \bigr\}
        \quad\text{and}\quad
        B_i(r) := \min_{i < b \le k} \bigl\{ \dist_H(r, v_b) + k - b \bigr\}.
    \end{gathered}
\end{equation*}
The second one, the \emph{block candidate}, collects the detours that stay inside the block $V_{q(i)}$ containing $v_i$:
\begin{equation*}\label{eq:block-candidate}
    S_i := \min_{\substack{0 \le a \le i < b \le k\\ v_a, v_b \in V_{q(i)}}} \bigl\{ a + \dist_{H_{q(i)}}(v_a, v_b) + k - b \bigr\}.
\end{equation*}
Note that, for $a \le i < b$, the term $\dist_{H_{q(a)}}(v_a, v_b)$ in \cref{lem:rp-decomp} is finite only if $v_b \in V_{q(a)}$, and then $q(a) = q(i) = q(b)$ because every block consists of consecutive layers.
Hence we obtain
\begin{equation*}
    \dist_{G \setminus e_i}(s, t) = \min\{ C_i, S_i \}
    \qquad \text{for every } 0 \le i < k,
\end{equation*}
and it remains to compute all values $C_i$ and $S_i$.

\cref{alg:rp} gives the pseudocode of our algorithm.
The pivot candidates are computed from the tables $A_i(r)$ and $B_i(r)$, which a prefix and a suffix scan along $\pi$ build for each pivot $r$.
The block candidates are computed in the same manner.
For a block $V_q$ and a vertex $v_a \in V(\gamma_q)$, a suffix scan along $\gamma_q$ builds the table
\begin{equation*}
    W_{q,a}(i) := \min_{i < b \le \min\{h_q, k\}} \bigl\{ \dist_{H_q}(v_a, v_b) + k - b \bigr\}
    \qquad \text{for } a \le i < \min\{h_q, k\},
\end{equation*}
and then $S_i = \min_{h_{q(i)-1} \le a \le i} \{ a + W_{q(i),a}(i) \}$, since $v_a \in V_{q(i)}$ with $a \le i$ means $h_{q(i)-1} \le a \le i$, and $v_b \in V_{q(i)}$ with $b > i$ means $i < b \le \min\{h_{q(i)}, k\}$.

\begin{algorithm}[!htb]
    \caption{The deterministic algorithm for RP.}\label{alg:rp}
    \Input{An unweighted directed graph $G = (V, E)$ and vertices $s, t \in V$.}
    \Output{$\dist_{G \setminus e_i}(s, t)$ for every edge $e_i = (v_i, v_{i+1})$ of $\pi$.}
    \BlankLine
    $L \gets \lceil \sqrt{n} \rceil$, and run \cref{alg:rp-blocks,alg:dist_compute}\;\label{line:rp-pre}
    \BlankLine
    \tcp{Pivot candidates $C_i$}
    \ForEach{$r \in \mathcal{R}$}{\label{line:rp-pivot-begin}
        $A \gets \infty$\;\label{line:rp-scanA-begin}
        \For{$i \gets 0$ \KwTo $k - 1$}{
            $A \gets \min\{A,\ i + \dist_H(v_i, r)\}$\;
            $A_i(r) \gets A$\label{line:rp-scanA-end}
        }
        $B \gets \infty$\;\label{line:rp-scanB-begin}
        \For{$i \gets k - 1$ \KwDownTo $0$}{
            $B \gets \min\{B,\ \dist_H(r, v_{i+1}) + k - (i + 1)\}$\;
            $B_i(r) \gets B$\label{line:rp-scanB-end}
        }
    }
    $C_i \gets \min_{r \in \mathcal{R}} \{A_i(r) + B_i(r)\}$ for all $0 \le i < k$\;\label{line:rp-pivot-end}
    \BlankLine
    \tcp{Block candidates $S_i$}
    \ForEach{block $V_q$}{\label{line:rp-block-begin}
        \ForEach{$v_a \in V(\gamma_q)$}{
            $W \gets \infty$\;\label{line:rp-scanW-begin}
            \For{$i \gets \min\{h_q, k\} - 1$ \KwDownTo $a$}{
                $W \gets \min\{W,\ \dist_{H_q}(v_a, v_{i+1}) + k - (i + 1)\}$\;
                $W_{q,a}(i) \gets W$\label{line:rp-scanW-end}
            }
        }
        $S_i \gets \min_{h_{q-1} \le a \le i} \{a + W_{q,a}(i)\}$ for all $i$ with $q(i) = q$\;\label{line:rp-block-end}
    }
    \BlankLine
    \Return $\min\{C_i, S_i\}$ for all $0 \le i < k$\;\label{line:rp-return}
\end{algorithm}

\paragraph*{Running time.}
\cref{alg:rp-blocks,alg:dist_compute} take $O(m)$ time and $O(mL + mn/L)$ time, respectively.
For the pivot candidates, each of the $|\mathcal{R}| \le n/L$ pivots performs two scans of $\pi$, each in $O(k) = O(n)$ time, which is $O(n^2/L)$ in total.
For the block candidates, each block contains at most $L+1$ vertices of $\pi$ and is processed in $O(L^2)$ time; since $\pi$ meets at most $k/L + 1 = O(n/L)$ blocks, this is $O(nL)$ in total.
Therefore, the total running time is $O(mn/L + mL + n^2/L + nL)$.
By choosing $L = \lceil \sqrt{n} \rceil$, we obtain $O(m\sqrt{n})$ time, which proves \cref{thm:rp}.

\section{Our Improved Combinatorial Algorithm for 2-FRP}\label{sec:overview-2frp}
In this section, we present our improved combinatorial algorithm for the 2-FRP problem.

We divide the cases into two based on the position of the two failed edges $e, e'$ as follows:
\begin{enumerate}
    \item One failure on $\pi$: $e \in E(\pi), e' \notin E(\pi)$
    \item Both failures on $\pi$: $e, e' \in E(\pi)$
\end{enumerate}
The first case is discussed in \cref{subsec:one_failure}, and the second case in \cref{subsec:both_failures}.
In both cases, we use the modulo-layer hitting set and the modulo-layer decomposition (\cref{subsec:modulo_layer}), which is the key to making our algorithm faster and simpler than that of Chechik and Zhang~\cite{CZ24}.

Note that we do not consider the case of $e, e' \notin E(\pi)$ because $\pi$ is the shortest in this situation.
Also, in the first case, we only consider $e'$ on the shortest $s$--$t$ path in $G \setminus e$.
More precisely, for $e \in E(\pi)$, let $\rho_e$ be a fixed shortest $s$--$t$ path in $G \setminus e$, that is, a replacement path for $e$.
Then, we consider $e' \in E(\rho_e) \setminus E(\pi)$.

\subsection{One Failure on $\pi$} \label{subsec:one_failure}
In the previous work by Chechik and Zhang~\cite{CZ24}, the running time in this case is $\tilde{O}(n^{3 - 1/18})$.
In this paper, we give an algorithm in $\tilde{O}(mn^{3/4} + n^{5/2})$ time for this case.

\begin{lemma}[One failure on $\pi$]\label{thm:oneoff}
    All values of $\dist_{G \setminus \{e, e'\}} (s, t)$ with $e \in E(\pi)$ and $e' \notin E(\pi)$ can be computed with high probability in $\Ot(mn^{3/4} + n^{5/2})$ time.
\end{lemma}

Throughout this subsection, we set $L=\lceil \sqrt{n} \rceil$.
Our algorithm for this case has two ingredients: the layer structure of \cref{sec:overview-rp} and the recent combinatorial algorithm of Harada and Izumi~\cite{HI26} for multiple source replacement paths (MSRP) problem.

We can use the modulo-layer hitting set and the modulo-layer decomposition (\cref{subsec:modulo_layer}) without any change.
Roughly speaking, since $e' \notin E(\pi)$, the only difference from the single-failure setting is that we need to search for detours in $H \setminus e'$ instead of $H$.
This change affects neither the layer structure nor its key property, since the separator layers, $\mathcal{R}$, and the blocks $V_q$ depend only on the distances from $s$ in $G$, and $\dist_G(s, y) \le \dist_G(s, x) + 1$ holds for every edge $(x, y)$ of $H \setminus e'$.
Hence, for every $e'$ and every $0 \le a < b \le k$, every $v_a$--$v_b$ path in $H \setminus e'$ either contains a vertex of $\mathcal{R}$ or lies entirely in $H_{q(a)} \setminus e'$.
In other words, we do not need to construct a hitting set separately for each failure $e' \notin E(\pi)$; the single set $\mathcal{R}$ serves as a hitting set for all of them simultaneously.
This allows us to avoid the extensive case analysis of Chechik and Zhang~\cite{CZ24}, which makes our algorithm for this case considerably simpler than theirs.

Formally, \cref{lem:rp-detour,lem:rp-decomp} hold verbatim with $H$ and $H_q$ replaced by $H \setminus e'$ and $H_q \setminus e'$, as stated below.

\begin{lemma}\label{lem:2frp-detour}
    For every edge $e = (v_i, v_{i+1}) \in \pi$ and every edge $e' \notin E(\pi)$,
    \[
        \dist_{G \setminus \{e, e'\}}(s, t) = \min_{0 \le a \le i < b \le k} \bigl\{ \dist_\pi(s, v_a) + \dist_{H \setminus e'}(v_a, v_b) + \dist_\pi(v_b, t) \bigr\}.
    \]
\end{lemma}

\begin{lemma}\label{lem:2frp-decomp}
    For each $0 \le a < b \le k$ and each edge $e' \notin E(\pi)$, with $q = q(a)$,
    \[
        \dist_{H \setminus e'}(v_a, v_b) = \min \Bigl\{ \min_{r \in \mathcal{R}} \bigl\{ \dist_{H \setminus e'}(v_a, r) + \dist_{H \setminus e'}(r, v_b) \bigr\},\ \dist_{H_q \setminus e'}(v_a, v_b) \Bigr\}.
    \]
\end{lemma}

\begin{proof}[Proof of \cref{lem:2frp-detour,lem:2frp-decomp}]
    The two lemmas are the counterparts of \cref{lem:rp-detour,lem:rp-decomp} with $H$ replaced by $H \setminus e'$, and their proofs are almost identical to those given there; we therefore only indicate the differences and omit the details.
    For \cref{lem:2frp-detour}, the subpaths of $\pi$ avoid $e'$ because $e' \notin E(\pi)$, and the middle part of a shortest $s$--$t$ path in $G \setminus \{e, e'\}$ between its last vertex $v_a$ with $a \le i$ and its first vertex $v_b$ with $b > i$ lies in $H \setminus e'$, since it has no internal vertex on $\pi$ and avoids $e'$.
    For \cref{lem:2frp-decomp}, a shortest $v_a$--$v_b$ path in $H \setminus e'$ is in particular a path of $H$, so the argument on layers applies verbatim; its subpaths are shortest paths in $H \setminus e'$, and when all of its vertices lie in $V_q$, it is a path of $H_q \setminus e'$.
\end{proof}

By \cref{lem:2frp-detour,lem:2frp-decomp}, all values $\dist_{G \setminus \{e, e'\}}(s, t)$ with $e \in E(\pi)$ and $e' \notin E(\pi)$ are determined by the distances $\dist_{H \setminus e'}(r, \ast)$ and $\dist_{H \setminus e'}(\ast, r)$ for $r \in \mathcal{R}$ and $\dist_{H_q \setminus e'}(v_i, \ast)$ for $v_i \in V_q$, over all $e'$.
To compute these three distances, we solve the instance of the multiple source replacement paths (MSRP) problem of \cref{subsec:msrp}, for which Harada and Izumi~\cite{HI26} give a combinatorial $\Ot(m\sqrt{\sigma n} + \sigma n^2)$-time algorithm with $\sigma$ sources (\cref{thm:msrp}).
Each of the three families above is exactly the output of one MSRP instance: for a graph $G'$ and a set $S$ of sources, \MSRP{$G'$, $S$} returns $\dist_{G' \setminus e'}(u, x)$ for every $u \in S$, every vertex $x$, and every edge $e'$ of $G'$ (for $e'$ off the fixed shortest-path tree of $u$ the distance is $\dist_{G'}(u, x)$).
\Cref{alg:dist_compute_msrp} is therefore the counterpart of \cref{alg:dist_compute} in which every BFS is replaced by an MSRP computation.
The sources of the first two instances are the $|\mathcal{R}| \le n/L$ vertices of $\mathcal{R}$, and the instance for a block $V_q$ has at most $L+1$ sources, so by \cref{thm:msrp} the total time is
$\Ot\bigl(mn/\sqrt{L} + n^3/L\bigr)$ for lines~\ref{line:msrp-pivot-out} and~\ref{line:msrp-pivot-in} and $\Ot\bigl(m\sqrt{Ln} + Ln^2\bigr)$ for line~\ref{line:msrp-block}.
Choosing $L = \lceil \sqrt{n} \rceil$ gives $\Ot(mn^{3/4} + n^{5/2})$, which is the bound of \cref{thm:oneoff}.

\begin{algorithm}[!htb]
    \caption{Preprocessing distance computation with one failure off $\pi$}\label{alg:dist_compute_msrp}
    $\dist_{H \setminus e'}(r, \ast) \gets \MSRP(H, \mathcal{R})$ for each $r \in \mathcal{R}$ and each $e' \in E(H)$\;\label{line:msrp-pivot-out}
    $\dist_{H \setminus e'}(\ast, r) \gets \MSRP(H^{\mathrm{rev}}, \mathcal{R})$ for each $r \in \mathcal{R}$ and each $e' \in E(H)$\;\label{line:msrp-pivot-in}
    $\dist_{H_q \setminus e'}(v, \ast) \gets \MSRP(H_q, V(\gamma_q))$ for each $q \ge 0, v \in V(\gamma_q)$, and $e' \in E(H_q)$\;\label{line:msrp-block}
\end{algorithm}

Once these distances are available, the overall algorithm is the same as that of \cref{sec:overview-rp}: for each failed edge $e \in E(\pi)$, the value $\dist_{G \setminus \{e, e'\}}(s, t)$ is the minimum of a pivot candidate and a block candidate, both of which are evaluated by prefix and suffix minima along $\pi$.
The only difference is in the treatment of the second failed edge $e'$, which lies off $\pi$: the two candidates are defined with the distances in $H \setminus e'$ and $H_q \setminus e'$ in place of those in $H$ and $H_q$, and they are evaluated for every $e' \in E(\rho_e) \setminus E(\pi)$.
The overall algorithm and its running time are described in \cref{app:oneoff}.

\subsection{Both Failures on $\pi$} \label{subsec:both_failures}
In the previous work by Chechik and Zhang~\cite{CZ24}, the running time in this case is $\tilde{O}(n^{3 - 1/7})$.
In this paper, we propose an algorithm in $\tilde{O}(n^{11/4})$ time for this case.

\begin{lemma}[Both failures on $\pi$]\label{lem:bothpath}
    All values of $\dist_{G \setminus \{ e, e' \}} (s, t)$ with $e, e' \in E(\pi)$ can be computed with high probability in $\Ot(n^{3 - 1/4})$ time.
\end{lemma}

In this case, we also use the modulo-layer hitting set and the modulo-layer decomposition (\cref{subsec:modulo_layer}) without any change.
As in \cref{sec:overview-rp} and \cref{alg:rp-blocks}, we partition $\pi$ into consecutive blocks $\gamma_0,\gamma_1,\ldots$. These blocks form a partition of $E(\pi)$, each block has length $O(L)$, and there are $O(n/L)$ blocks in total. 
Throughout this subsection, we set $L=\lceil n^{3/4} \rceil$.

For each pair of failed edges, we consider the blocks that contain the two failures.
For $e\in E(\gamma_p)$ and $e'\in E(\gamma_q)$ with $p\le q$, we distinguish whether the two failed edges lie in the same block or in different blocks, and handle the two cases separately.
The former case ($p = q$) is discussed in \cref{subsubsec:same_block}, and the latter case ($p < q$) in \cref{subsubsec:different}.

\paragraph*{Rough overview of the case where both failures lie on $\pi$.}
For edge pairs in the same block, we construct for each block an auxiliary weighted directed graph with $O(L)$ vertices and solve 2-FRP on this graph using \cref{thm:weighted-2frp}. Since there are $O(n/L)$ blocks and each instance can be solved in $\tO(L^3)$ time, the total time for this computation is $\tO(nL^2) = \tO(n^{5/2})$.
For edge pairs in different blocks, we reduce the problem to $O(n/L) = O(n^{1/4})$ SSRP computations, which take $\tO(n^{3-1/4})$ time in total.

\subsubsection{Edge Pairs in the Same Block} \label{subsubsec:same_block}
We consider two failures $e, e' \in E(\pi)$ that lie on the same block, that is, $e_i, e_j \in E(\gamma_q)$ for some $q$.
Our approach builds a \emph{weighted kernel} on the vertex set $V(\gamma_q) \cup \{s', t'\}$, where $s'$ and $t'$ are two new vertices, so that $K_q$ has at most $L+3$ vertices and $\dist_{K_q \setminus \{e, e'\}}(s', t') = \dist_{G \setminus \{e, e'\}}(s, t)$ for every pair $e, e' \in E(\gamma_q)$.

\paragraph*{Kernel construction.}\label{para:kernel-construction}
The kernel $K_q$ is obtained by adding the following weighted edges (see \cref{fig:kernel}):
\begin{enumerate}[label=Type $\arabic*$:, leftmargin=*]
    \item every edge of $\gamma_q$ of weight $1$
    \item an edge $(s', v)$ of weight $\dist_{G \setminus E(\gamma_q)}(s, v)$ for each $v \in V(\gamma_q)$
    \item an edge $(v, t')$ of weight $\dist_{G \setminus E(\gamma_q)}(v, t)$ for each $v \in V(\gamma_q)$
    \item an edge $(v_a, v_b)$ of weight $\dist_H(v_a, v_b)$ for each $v_a, v_b \in V(\gamma_q)$
    \item an edge $(s', t')$ of weight $\dist_{G \setminus E(\gamma_q)}(s, t)$
\end{enumerate}
Edges of infinite weight are omitted, and an edge of Type 4 is kept distinct from a parallel edge of $\gamma_q$.
The following lemma shows that the kernel is correct.

\begin{figure}[!tb]
    \centering
    \begin{subfigure}[b]{0.49\textwidth}
        \centering
        \includegraphics[width=\linewidth]{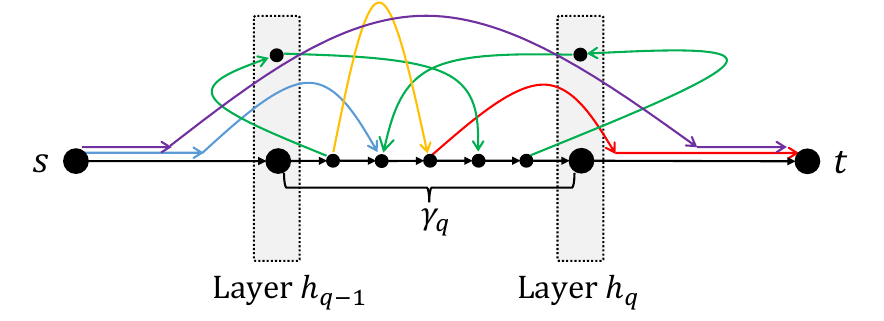}
        \caption{The block $\gamma_q$ in $G$}\label{fig:kernel-graph}
    \end{subfigure}\hfill
    \begin{subfigure}[b]{0.49\textwidth}
        \centering
        \includegraphics[width=\linewidth]{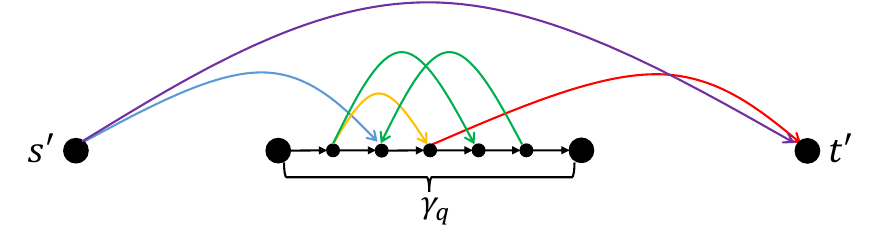}
        \caption{The kernel $K_q$ on $V(\gamma_q) \cup \{s', t'\}$.}\label{fig:kernel-kernel}
    \end{subfigure}
    \caption{Construction of the kernel $K_q$ for edge pairs in the same block.
    In (a), the black path is $\pi$ and $\gamma_q$ is its subpath between the separator layers $h_{q-1}$ and $h_q$ (shaded).
    Each colored path of $G$ becomes a kernel edge of the same length in (b):
    the blue path from $s$, the red path to $t$, and the purple $s$--$t$ path avoid $E(\gamma_q)$ (Types 2, 3, and 5), and the yellow and green detours are the two terms of \cref{lem:block-detour} (Type 4), where a yellow detour stays inside $H_q$ and a green detour passes through a pivot on a separator layer, in either direction along $\pi$.
    In (b), only the edges of $\gamma_q$ (weight one) can fail, and only a few weighted edges are drawn.}\label{fig:kernel}
\end{figure}

\begin{lemma}\label{lem:kernel-correctness}
    For each $q \ge 0$ and each two failed edges $e, e' \in E(\gamma_q)$, 
    \[\dist_{K_q \setminus \{e, e'\}}(s', t') = \dist_{G \setminus \{e, e'\}}(s, t).\]
\end{lemma}

\begin{proof}
    Every edge of $K_q \setminus \{e, e'\}$ expands to a walk of the same length in $G \setminus \{e, e'\}$. Therefore, $\dist_{K_q \setminus \{e, e'\}}(s', t')$ is at least $\dist_{G \setminus \{e, e'\}}(s, t)$.

    For the converse inequality, let $\rho$ be a canonical shortest $s$--$t$ path of $G \setminus \{e, e'\}$ with respect to $\pi$ and $\{e, e'\}$ (\cref{def:canonical}).
    It suffices to find an $s'$--$t'$ walk of $K_q \setminus \{e, e'\}$ of length at most $|\rho|$.
    If $\rho$ uses no edge of $\gamma_q$, then $\rho$ is a path of $G \setminus E(\gamma_q)$, and the edge $(s', t')$ of Type 5 has weight at most $|\rho|$.

    Otherwise, let $(v_i, v_{i+1})$ and $(v_{i'}, v_{i'+1})$ be the first and the last edge of $\gamma_q$ on $\rho$, and split $\rho$ into the prefix $\rho[s, v_i]$, the middle part $\rho[v_i, v_{i'+1}]$, and the suffix $\rho[v_{i'+1}, t]$.
    The prefix and the suffix use no edge of $\gamma_q$, so they are paths of $G \setminus E(\gamma_q)$, and their lengths are at least the weights of the edge $(s', v_i)$ of Type 2 and the edge $(v_{i'+1}, t')$ of Type 3, respectively.
    Since $\rho$ is canonical and both failures lie on $\gamma_q$, the middle part uses no edge of $\pi$ outside $\gamma_q$.
    Hence the middle part decomposes into edges of $\gamma_q$ and maximal subpaths of $H$.
    Every edge of $\gamma_q$ on $\rho$ is different from $e$ and $e'$, so it is an edge of Type 1 of $K_q \setminus \{e, e'\}$.
    Every maximal subpath of $H$ starts at $v_i$ or at the head of an edge of $\gamma_q$ and ends at $v_{i'+1}$ or at the tail of an edge of $\gamma_q$; hence it connects two vertices $v_a, v_b \in V(\gamma_q)$, and its length is at least $\dist_H(v_a, v_b)$, the weight of the edge $(v_a, v_b)$ of Type 4.
    Replacing the prefix, the suffix, and every piece of the middle part by the corresponding kernel edge yields an $s'$--$t'$ walk of $K_q \setminus \{e, e'\}$ of length at most $|\rho|$, which proves the converse inequality.
\end{proof}

By the above lemma, the algorithm of~\cref{thm:weighted-2frp} answers all pairs of the block in $\tilde{O}(L^3)$ time, and all $O(n/L)$ blocks in $\tO(nL^2)$ time.

To construct the kernels, their edge weights have to be computed in a preprocessing step.
The weights of Types 2, 3, and 5 are obtained by one forward and one reverse BFS in $G \setminus E(\gamma_q)$ per block, in $O(mn/L)$ time in total.
The delicate part is the weights of Type 4, that is, the detour distances $\dist_H(v_a, v_b)$ for $v_a, v_b \in V(\gamma_q)$.
A canonical replacement path for dual failures may also contain a \emph{backward detour}, from $v_a$ to $v_b$ with $b < a$.
Backward detours are not intercepted by $\mathcal{R}$: the argument of \cref{lem:rp-decomp} relies on the fact that a path climbing from layer $a$ to layer $b > a$ visits every layer in between, whereas a path descending from layer $a$ to layer $b < a$ may skip all separator layers in between.
This is why we restrict the kernel to pairs inside one block: a detour between two vertices of the same block either stays within the layers of the block or leaves them and must climb back in through a separator layer, so the decomposition of \cref{lem:rp-decomp} survives regardless of the order of $a$ and $b$. 

\begin{lemma}\label{lem:block-detour}
    For each $q \ge 0$ and each $h_{q-1} \le a, b \le h_q$,
    \[
        \dist_H(v_a, v_b) = \min \Bigl\{ \min_{r \in \mathcal{R}} \bigl\{ \dist_H(v_a, r) + \dist_H(r, v_b) \bigr\},\ \dist_{H_q}(v_a, v_b) \Bigr\}.
    \]
\end{lemma}

\begin{proof}
    The proof of \cref{lem:rp-decomp} applies verbatim: it uses the order of $a$ and $b$ only to guarantee $\dist_G(s, v_b) \ge h_{q-1}$, which now holds by assumption.
\end{proof}

The distances $\dist_H(v_a, r)$ and $\dist_H(r, v_b)$ for all pivots $r$ are obtained by a forward and a reverse BFS in $H$ from each of the $|\mathcal{R}| \le n/L$ pivots, in $O(mn/L)$ time, and the distances $\dist_{H_q}(v_a, v_b)$ by a BFS in $H_q$ from each of the at most $L+1$ vertices of $\gamma_q$; since every vertex and every edge of $H$ lies in at most two of the graphs $H_q$, these BFSs cost $O(mL)$ time in total.
Evaluating the minimum of \cref{lem:block-detour} for the at most $L^2$ pairs of a block takes $O(L^2 \cdot n/L) = O(nL)$ time, hence $O(n^2)$ time over all $O(n/L)$ blocks.
Together with the $\tO(nL^2)$ time for running the algorithm of \cref{thm:weighted-2frp} on every kernel, all edge pairs in the same block are handled deterministically in $\tO(mL + mn/L + n^2 + nL^2) = \tO(mn^{3/4} + n^{5/2})$ time, recalling that $L=\lceil n^{3/4} \rceil$.
The whole procedure is summarized in \cref{alg:same-block}.

\begin{algorithm}[!htb]
    \caption{Edge pairs in the same block}\label{alg:same-block}
    $\dist_H(r, \ast) \gets \BFS(H, r)$ for each $r \in \mathcal{R}$\;
    $\dist_H(\ast, r) \gets \BFS(H^{\mathrm{rev}}, r)$ for each $r \in \mathcal{R}$\;
    \ForEach{block $q$}{
        $\dist_{G \setminus E(\gamma_q)}(s, \ast) \gets \BFS(G \setminus E(\gamma_q), s)$\;
        $\dist_{G \setminus E(\gamma_q)}(\ast, t) \gets \BFS((G \setminus E(\gamma_q))^{\mathrm{rev}}, t)$\;
        $\dist_{H_q}(v, \ast) \gets \BFS(H_q, v)$ for each $v \in V(\gamma_q)$\;
        \ForEach{$v_a, v_b \in V(\gamma_q)$}{
            $\dist_H(v_a, v_b) \gets \min \Bigl\{ \min_{r \in \mathcal{R}} \bigl\{ \dist_H(v_a, r) + \dist_H(r, v_b) \bigr\},\ \dist_{H_q}(v_a, v_b) \Bigr\}$   
        }
        Construct the weighted kernel $K_q$ by the procedure described in \hyperref[para:kernel-construction]{Kernel construction}\;
        Run the algorithm of \cref{thm:weighted-2frp} for $K_q$ and $s', t'$.
    }
\end{algorithm}

\subsubsection{Edge Pairs in Different Blocks} \label{subsubsec:different}
Fix two distinct blocks $\gamma_p=\pi[v_a,v_b]$ and $\gamma_q=\pi[v_c,v_d]$, where $p+1 \le q$ and hence $b \le c$.
Let $e=(v_i,v_{i+1})\in E(\gamma_p)$ and $e'=(v_j,v_{j+1})\in E(\gamma_q)$, so that $a\le i<b \le c\le j<d$.

Let $\rho$ be a canonical shortest $s$--$t$ path in $G\setminus\{e,e'\}$.
Such a path always exists.
We distinguish four cases according to how $\rho$ traverses the portions of $\pi$ around the two failed edges.
See \cref{fig:four_cases} for an illustration.

\begin{enumerate}[label=Case $\arabic*$:, leftmargin=*]
    \item $\rho$ does not use any edge of $\pi[v_{i+1},v_b]$.

    \item $\rho$ does not use any edge of $\pi[v_c,v_j]$.

    \item Neither Case~1 nor Case~2 occurs.
    Let $x$ be the minimum index such that $v_x\in V(\pi[v_{i+1},v_b])\cap V(\rho)$, and let $y$ be the maximum index such that $v_y\in V(\pi[v_c,v_j])\cap V(\rho)$.
    Note that such vertices $v_x$ and $v_y$ necessarily exist, since neither Case~1 nor Case~2 occurs.
    Suppose that $v_x$ appears at or before $v_y$ on $\rho$.
    Since $x\le b\le c\le y$, we have $v_x=v_y$ only if $b=c$ and $x=y=b$, which can happen when $q=p+1$.

    \item Neither Case~1 nor Case~2 occurs, and, with $v_x$ and $v_y$ defined as in Case~3, $v_y$ appears strictly before $v_x$ on $\rho$.
\end{enumerate}

\begin{figure}[t]
    \centering

    \begin{minipage}[t]{0.47\linewidth}
    \centering
    \begin{tikzpicture}[x=0.46cm,y=0.58cm, every node/.style={font=\scriptsize}]

        \fill (-0.7,0) circle (1.1pt); \node[above] at (-0.7,0) {$s$};
        \fill (1.6,0) circle (1.1pt);  \node[above] at (1.6,0) {$v_i$};
        \fill (2.5,0) circle (1.1pt);  \node[above] at (2.5,0) {$v_{i+1}$};
        \fill (3.8,0) circle (1.1pt);  \node[above] at (3.8,0) {$v_b$};
        \fill (5.6,0) circle (1.1pt);  \node[above] at (5.6,0) {$v_c$};
        \fill (6.9,0) circle (1.1pt);  \node[above] at (6.9,0) {$v_j$};
        \fill (7.8,0) circle (1.1pt);  \node[above] at (7.8,0) {$v_{j+1}$};
        \fill (10.1,0) circle (1.1pt); \node[above] at (10.1,0) {$t$};

        \draw[thick] (-0.7,0) -- (10.1,0);

        \draw[red,very thick] (1.6,0) -- (2.5,0);
        \draw[red,very thick] (1.92,-0.14) -- (2.18,0.14);
        \draw[red,very thick] (1.92,0.14) -- (2.18,-0.14);
        \node[red,below] at (2.05,-0.12) {$e$};

        \draw[red,very thick] (6.9,0) -- (7.8,0);
        \draw[red,very thick] (7.22,-0.14) -- (7.48,0.14);
        \draw[red,very thick] (7.22,0.14) -- (7.48,-0.14);
        \node[red,below] at (7.35,-0.12) {$e'$};

        \fill[blue] (1.0,0) circle (1.4pt);
        \draw[blue,very thick] (-0.7,0) -- (1.0,0);
        \draw[blue,very thick,->]
            (1.0,0) .. controls (3.0,1.8) and (7.8,1.8) .. (10.1,0);

        \node at (4.7,-1.15) {\textbf{Case 1}};
        \node at (4.7,-1.55) {$\rho$ contains no edge of $\pi[v_{i+1},v_b]$.};
    \end{tikzpicture}
    \end{minipage}
    \hfill
    \begin{minipage}[t]{0.47\linewidth}
    \centering
    \begin{tikzpicture}[x=0.46cm,y=0.58cm, every node/.style={font=\scriptsize}]

        \fill (-0.7,0) circle (1.1pt); \node[above] at (-0.7,0) {$s$};
        \fill (1.6,0) circle (1.1pt);  \node[above] at (1.6,0) {$v_i$};
        \fill (2.5,0) circle (1.1pt);  \node[above] at (2.5,0) {$v_{i+1}$};
        \fill (3.8,0) circle (1.1pt);  \node[above] at (3.8,0) {$v_b$};
        \fill (5.6,0) circle (1.1pt);  \node[above] at (5.6,0) {$v_c$};
        \fill (6.9,0) circle (1.1pt);  \node[above] at (6.9,0) {$v_j$};
        \fill (7.8,0) circle (1.1pt);  \node[above] at (7.8,0) {$v_{j+1}$};
        \fill (10.1,0) circle (1.1pt); \node[above] at (10.1,0) {$t$};

        \draw[thick] (-0.7,0) -- (10.1,0);

        \draw[red,very thick] (1.6,0) -- (2.5,0);
        \draw[red,very thick] (1.92,-0.14) -- (2.18,0.14);
        \draw[red,very thick] (1.92,0.14) -- (2.18,-0.14);
        \node[red,below] at (2.05,-0.12) {$e$};

        \draw[red,very thick] (6.9,0) -- (7.8,0);
        \draw[red,very thick] (7.22,-0.14) -- (7.48,0.14);
        \draw[red,very thick] (7.22,0.14) -- (7.48,-0.14);
        \node[red,below] at (7.35,-0.12) {$e'$};

        \fill[blue] (8.45,0) circle (1.4pt);
        \draw[blue,very thick,->]
            (-0.7,0) .. controls (1.8,1.95) and (6.0,1.95) .. (8.45,0);
        \draw[blue,very thick] (8.45,0) -- (10.1,0);

        \node at (4.7,-1.15) {\textbf{Case 2}};
        \node at (4.7,-1.55) {$\rho$ contains no edge of $\pi[v_c,v_j]$.};
    \end{tikzpicture}
    \end{minipage}

    \vspace{1.5em}

    \begin{minipage}[t]{0.47\linewidth}
    \centering
    \begin{tikzpicture}[x=0.46cm,y=0.58cm, every node/.style={font=\scriptsize}]

        \fill (-0.7,0) circle (1.1pt); \node[above] at (-0.7,0) {$s$};
        \fill (1.6,0) circle (1.1pt);  \node[above] at (1.6,0) {$v_i$};
        \fill (2.5,0) circle (1.1pt);  \node[above] at (2.5,0) {$v_{i+1}$};
        \fill (3.8,0) circle (1.1pt);  \node[above] at (3.8,0) {$v_b$};
        \fill (5.6,0) circle (1.1pt);  \node[above] at (5.6,0) {$v_c$};
        \fill (6.9,0) circle (1.1pt);  \node[above] at (6.9,0) {$v_j$};
        \fill (7.8,0) circle (1.1pt);  \node[above] at (7.8,0) {$v_{j+1}$};
        \fill (10.1,0) circle (1.1pt); \node[above] at (10.1,0) {$t$};

        \draw[thick] (-0.7,0) -- (10.1,0);

        \draw[red,very thick] (1.6,0) -- (2.5,0);
        \draw[red,very thick] (1.92,-0.14) -- (2.18,0.14);
        \draw[red,very thick] (1.92,0.14) -- (2.18,-0.14);
        \node[red,below] at (2.05,-0.12) {$e$};

        \draw[red,very thick] (6.9,0) -- (7.8,0);
        \draw[red,very thick] (7.22,-0.14) -- (7.48,0.14);
        \draw[red,very thick] (7.22,0.14) -- (7.48,-0.14);
        \node[red,below] at (7.35,-0.12) {$e'$};

        \fill[blue] (3.1,0) circle (1.5pt);
        \node[below] at (3.1,-0.26) {$v_x$};

        \fill[blue] (6.3,0) circle (1.5pt);
        \node[below] at (6.3,-0.26) {$v_y$};

        \draw[blue,very thick,->]
            (-0.7,0) .. controls (1.0,1.2) and (2.2,1.2) .. (3.1,0);
        \draw[blue,very thick] (3.1,0) -- (6.3,0);
        \draw[blue,very thick,->]
            (6.3,0) .. controls (7.9,1.2) and (9.1,1.2) .. (10.1,0);

        \node at (4.7,-1.15) {\textbf{Case 3}};
        \node at (4.7,-1.55) {$v_x$ appears at or before $v_y$ on $\rho$.};
    \end{tikzpicture}
    \end{minipage}
    \hfill
    \begin{minipage}[t]{0.47\linewidth}
    \centering
    \begin{tikzpicture}[x=0.46cm,y=0.58cm, every node/.style={font=\scriptsize}]

        \fill (-0.7,0) circle (1.1pt); \node[above] at (-0.7,0) {$s$};
        \fill (1.6,0) circle (1.1pt);  \node[above] at (1.6,0) {$v_i$};
        \fill (2.5,0) circle (1.1pt);  \node[above] at (2.5,0) {$v_{i+1}$};
        \fill (3.8,0) circle (1.1pt);  \node[above] at (3.8,0) {$v_b$};
        \fill (5.6,0) circle (1.1pt);  \node[above] at (5.6,0) {$v_c$};
        \fill (6.9,0) circle (1.1pt);  \node[above] at (6.9,0) {$v_j$};
        \fill (7.8,0) circle (1.1pt);  \node[above] at (7.8,0) {$v_{j+1}$};
        \fill (10.1,0) circle (1.1pt); \node[above] at (10.1,0) {$t$};

        \draw[thick] (-0.7,0) -- (10.1,0);

        \draw[red,very thick] (1.6,0) -- (2.5,0);
        \draw[red,very thick] (1.92,-0.14) -- (2.18,0.14);
        \draw[red,very thick] (1.92,0.14) -- (2.18,-0.14);
        \node[red,below] at (2.05,-0.12) {$e$};

        \draw[red,very thick] (6.9,0) -- (7.8,0);
        \draw[red,very thick] (7.22,-0.14) -- (7.48,0.14);
        \draw[red,very thick] (7.22,0.14) -- (7.48,-0.14);
        \node[red,below] at (7.35,-0.12) {$e'$};

        \fill[blue] (3.1,0) circle (1.5pt);
        \node[below] at (3.1,-0.26) {$v_x$};

        \fill[blue] (6.3,0) circle (1.5pt);
        \node[below] at (6.3,-0.26) {$v_y$};

        \draw[blue,very thick,->]
            (-0.7,0) .. controls (1.8,1.65) and (4.3,1.65) .. (6.3,0);
        \draw[blue,very thick,->]
            (6.3,0) .. controls (5.6,-1.45) and (4.1,-1.45) .. (3.1,0);
        \draw[blue,very thick,->]
            (3.1,0) .. controls (5.6,1.0) and (7.8,1.0) .. (10.1,0);

        \node at (4.7,-1.95) {\textbf{Case 4}};
        \node at (4.7,-2.35) {$v_y$ appears strictly before $v_x$ on $\rho$.};
    \end{tikzpicture}
    \end{minipage}

    \caption{
        The four cases for a canonical shortest $s$--$t$ path $\rho$ in $G\setminus\{e,e'\}$.
        The blue curve represents $\rho$. 
        The figure is schematic, and $\rho$ may still intersect $\pi$ at other vertices.
    }
    \label{fig:four_cases}
\end{figure}
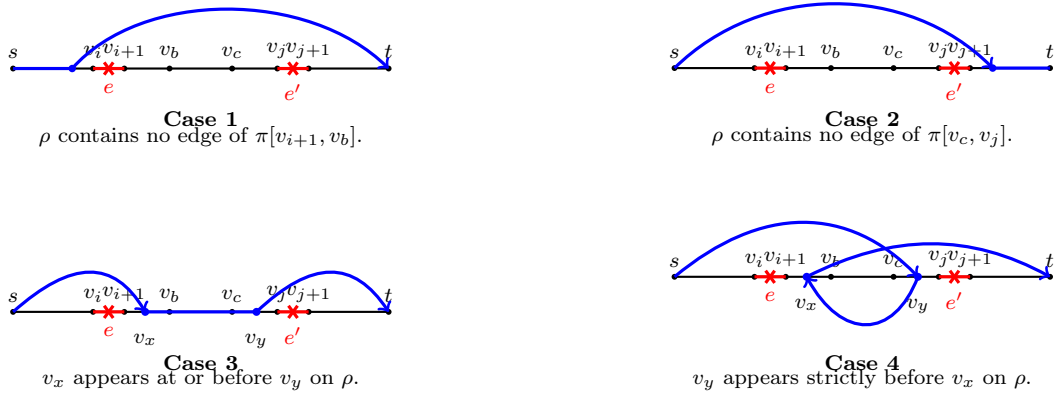

For every pair of edges $(e,e')$ in different blocks, let $D_{\mathrm L}(e,e')$, $D_{\mathrm R}(e,e')$, $D_{\mathrm F}(e,e')$, and $D_{\mathrm B}(e,e')$ denote  the minimum lengths of $s$--$t$ paths in $G\setminus \{e,e'\}$ satisfying the conditions of Cases~1, 2, 3, and 4, respectively.
Then, 
\[
    \dist_{G\setminus\{e,e'\}}(s,t)
    =
    \min\bigl\{
        D_{\mathrm{L}}(e,e'),
        D_{\mathrm{R}}(e,e'),
        D_{\mathrm{F}}(e,e'),
        D_{\mathrm{B}}(e,e')
    \bigr\}.
\]

We next describe how to compute these candidates efficiently.
Roughly speaking, we show that all of them can be obtained using only $O(n/L)$ SSRP computations.

\paragraph*{Case 1: $\rho$ contains no edge of $\pi[v_{i+1},v_b]$.}
Let $v_z$ be the last vertex of $\rho$ in $\{v_0,\ldots,v_i\}$. 
Since $\pi[s,v_z]$ contains no failed edge, the canonicality of $\rho$ implies that $\rho[s,v_z]=\pi[s,v_z]$, and hence $|\rho[s,v_z]|=z$. 
Moreover, $\rho[v_z,t]$ contains no edge of $\gamma_p$.
Indeed, by the choice of $v_z$, it cannot use an edge of $\pi[v_a,v_i]$, the edge $e$ is failed, and by the assumption of this case it uses no edge of $\pi[v_{i+1},v_b]$. 
Therefore, we have $|\rho| = z+\dist_{G\setminus(E(\gamma_p)\cup\{e'\})}(v_z,t)$.

Accordingly, define
\[
    D_{\mathrm L}(e,e')
    :=
    \min_{0\le z\le i}
    \left\{
        z+\dist_{G\setminus(E(\gamma_p)\cup\{e'\})}(v_z,t)
    \right\}.
\]
Every term in the minimum corresponds to an $s$--$t$ walk avoiding both $e$ and $e'$, while the choice of $v_z$ above attains $|\rho|$. 
Hence, $|\rho| = D_{\mathrm L}(e,e')$ in this case.

\subparagraph*{Computing $D_{\mathrm L}$.}
For each block $\gamma_p$, let $G_p:=G\setminus E(\gamma_p)$ and run SSRP from $t$ in $G_p^{\mathrm{rev}}$.
This gives, for every failed edge $e'$ and every vertex $v_z$, the distance $\dist_{G_p\setminus e'}(v_z,t)$.
For each fixed $e'$, we then scan the vertices $v_0,v_1,\ldots$ in their order on $\pi$ and maintain the prefix minima
$
    P_{e'}(i)
    :=
    \min_{0\le z\le i}
    \left\{
        z+\dist_{G_p\setminus e'}(v_z,t)
    \right\}.
$
Then, for every $e=(v_i,v_{i+1})\in E(\gamma_p)$, we simply have $D_{\mathrm L}(e,e')=P_{e'}(i)$.

\subparagraph*{Running time.}
There are $O(n/L)$ blocks. 
By Lemma~\ref{thm:msrp}, one call to SSRP takes $\Ot(m\sqrt{n}+n^2)$ time.
Hence, the total time for the SSRP calls over all blocks is
$
    \Ot\left(
        \frac{n}{L}(m\sqrt n+n^2)
    \right)
    =
    \Ot\left(
        (mn^{3/2}+n^3)/L
    \right).
$
Recall that there are only $O(n)$ possible choices for the second failed edge $e'$.
The prefix-minimum scans require at most $O(n^2)$ additional time per block, and hence $O(n^3/L)$ time over all blocks.
Consequently, all values $D_{\mathrm L}(e,e')$ for edge pairs $(e,e')$ lying in different blocks can be computed in $\displaystyle  \tO\left((mn^{3/2}+n^3)/L\right)$ time.

\paragraph*{Case 2: $\rho$ contains no edge of $\pi[v_c,v_j]$.}

This case is symmetric to Case~1.
Let $v_z$ be the first vertex of $\rho$ in $\{v_{j+1},\ldots,v_k \}$.
By canonicality, $\rho[v_z,t]=\pi[v_z,t]$, and hence $|\rho[v_z,t]|=k-z$.
Moreover, $\rho[s,v_z]$ contains no edge of $\gamma_q$.
Therefore, we define
\[
D_{\mathrm{R}}(e, e') := \min_{j < z \le k} \{ \dist_{G \setminus (E(\gamma_q) \cup \{ e \}) } (s, v_z) + (k - z) \}.
\]
Every term in the minimum corresponds to an $s$--$t$ walk avoiding both $e$ and $e'$, while the choice of $v_z$ above attains $|\rho|$. 
Hence, we have $|\rho| = D_{\mathrm{R}}(e, e')$ in Case~2.

To compute these values, for each block $\gamma_q$, we run SSRP from $s$ in $G\setminus E(\gamma_q)$.
For each fixed failed edge $e$, a suffix-minimum scan along $\pi$ then gives $D_{\mathrm R}(e,e')$ for every $e'=(v_j,v_{j+1})\in E(\gamma_q)$.
As in Case~1, all values $D_{\mathrm R}(e,e')$ can be computed in $\tO((mn^{3/2}+n^3)/L)$ time.

\paragraph*{Case 3: $v_x$ appears at or before $v_y$ on $\rho$.}
Since neither Case~1 nor Case~2 occurs, the vertices $v_x$ and $v_y$ are well defined.
Recall that $x$ is the minimum index such that $v_x\in V(\pi[v_{i+1},v_b])\cap V(\rho)$, and $y$ is the maximum index such that $v_y\in V(\pi[v_c,v_j])\cap V(\rho)$.
Thus, $i<x\le b\le c\le y\le j$.
Hence, $\pi[v_x,v_y]$ contains neither $e$ nor $e'$.
Since $v_x$ appears at or before $v_y$ on both $\pi$ and $\rho$, the canonicality of $\rho$ implies that $\rho[v_x,v_y]=\pi[v_x,v_y]$.

Moreover, $\rho[s,v_b]$ uses no edge of $\pi[v_b, t]$, and $\rho[v_c,t]$ uses no edge of $\pi[s, v_c]$.
Therefore, $\rho[s,v_b]$ uses only edges of $G\setminus E(\pi[v_b, t]) = H + \pi[s, v_b]$ and avoids $e$. 
Recall that $H := G \setminus E(\pi)$.
Similarly, $\rho[v_c,t]$ uses only edges of $G \setminus E(\pi[s, v_c]) = H + \pi[v_c,t]$ and avoids $e'$.

Accordingly, define
\[
D_{\mathrm F}(e,e')
:=
\dist_{(H+\pi[s,v_b])\setminus e}(s,v_b)
+(c-b)
+\dist_{(H+\pi[v_c,t])\setminus e'}(v_c,t).
\]
In Case~3, we have $|\rho|=D_{\mathrm F}(e,e')$.

The key observation is that $e'$ is absent from $H+\pi[s,v_b]$, while $e$ is absent from $H+\pi[v_c,t]$. 
Hence, each of the two distance terms in $D_{\mathrm F}(e,e')$ involves only one failed edge.

\subparagraph*{Computing $D_{\mathrm F}$ and its running time.}
For every block $\gamma_p=\pi[v_a,v_b]$, run SSRP from $s$ in $H+\pi[s,v_b]$.
This gives $\dist_{(H+\pi[s,v_b])\setminus e}(s,v_b)$ for every
$e\in E(\gamma_p)$.
Symmetrically, for every block $\gamma_q=\pi[v_c,v_d]$, run SSRP from $t$ in $(H+\pi[v_c,t])^{\mathrm rev}$.
This gives $\dist_{(H+\pi[v_c,t])\setminus e'}(v_c,t)$ for every $e'\in E(\gamma_q)$.
Once these values are available, $D_{\mathrm F}(e,e')$ can be obtained in constant time for every pair $(e,e')$.

There are $O(n/L)$ blocks, and each SSRP call takes $\Ot(m\sqrt n+n^2)$ time by Lemma~\ref{thm:msrp}.
Hence, all values $D_{\mathrm F}(e,e')$ can be computed in $\Ot((mn^{3/2}+n^3)/L)$ time, together with $O(n^2)$ additional time to compute $D_{\mathrm F}(e,e')$ for all edge pairs $(e, e')$.

\paragraph*{Case 4: $v_y$ appears strictly before $v_x$ on $\rho$.}

Let $S_p:=\{u\in V \mid \dist_G(s,u)=b\}$ and $T_q:=\{w\in V \mid \dist_G(s,w)=c\}$.
Let $u$ be the first vertex of $\rho$ in $S_p$, and let $w$ be the last vertex of $\rho$ in $T_q$.
See \cref{fig:case4} for an illustration.

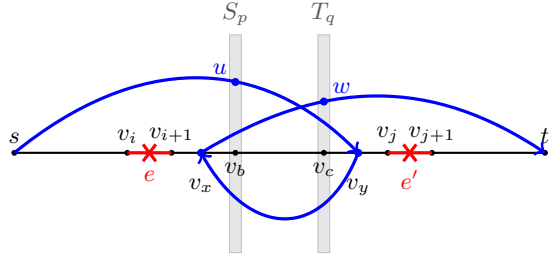
\begin{figure}[t]
\centering
\begin{tikzpicture}[x=0.65cm,y=0.80cm, every node/.style={font=\footnotesize}]

    \draw[draw=black!25,fill=black!10,line width=0.3pt]
        (3.68,-1.65) rectangle (3.92,1.95);
    \node[black!65,above] at (3.8,1.95) {$S_p$};
    \draw[draw=black!25,fill=black!10,line width=0.3pt]
        (5.48,-1.65) rectangle (5.72,1.95);
    \node[black!65,above] at (5.6,1.95) {$T_q$};

    \path[name path=layer-sp] (3.8,-1.65) -- (3.8,1.95);
    \path[name path=layer-tq] (5.6,-1.65) -- (5.6,1.95);

    \fill (-0.7,0) circle (1.1pt); \node[above] at (-0.7,0) {$s$};
    \fill (1.6,0) circle (1.1pt);  \node[above] at (1.6,0) {$v_i$};
    \fill (2.5,0) circle (1.1pt);  \node[above] at (2.5,0) {$v_{i+1}$};
    \fill (3.8,0) circle (1.1pt);  \node[below] at (3.8,0) {$v_b$};
    \fill (5.6,0) circle (1.1pt);  \node[below] at (5.6,0) {$v_c$};
    \fill (6.9,0) circle (1.1pt);  \node[above] at (6.9,0) {$v_j$};
    \fill (7.8,0) circle (1.1pt);  \node[above] at (7.8,0) {$v_{j+1}$};
    \fill (10.1,0) circle (1.1pt); \node[above] at (10.1,0) {$t$};

    \draw[thick] (-0.7,0) -- (10.1,0);

    \draw[red,very thick] (1.6,0) -- (2.5,0);
    \draw[red,very thick] (1.92,-0.14) -- (2.18,0.14);
    \draw[red,very thick] (1.92,0.14) -- (2.18,-0.14);
    \node[red,below] at (2.05,-0.12) {$e$};

    \draw[red,very thick] (6.9,0) -- (7.8,0);
    \draw[red,very thick] (7.22,-0.14) -- (7.48,0.14);
    \draw[red,very thick] (7.22,0.14) -- (7.48,-0.14);
    \node[red,below] at (7.35,-0.12) {$e'$};

    \fill[blue] (3.1,0) circle (1.5pt);
    \node[below] at (3.1,-0.26) {$v_x$};

    \fill[blue] (6.3,0) circle (1.5pt);
    \node[below] at (6.3,-0.26) {$v_y$};

    \draw[blue,very thick,->,name path=rho-first]
        (-0.7,0) .. controls (1.8,1.65) and (4.3,1.65) .. (6.3,0);
    \draw[blue,very thick,->]
        (6.3,0) .. controls (5.6,-1.45) and (4.1,-1.45) .. (3.1,0);
    \draw[blue,very thick,->,name path=rho-last]
        (3.1,0) .. controls (5.6,1.25) and (7.8,1.25) .. (10.1,0);

    \path[name intersections={of=rho-first and layer-sp,by=u}];
    \path[name intersections={of=rho-last and layer-tq,by=w}];
    \fill[blue] (u) circle (1.5pt);
    \node[blue,above left,inner sep=3pt] at (u) {$u$};

    \fill[blue] (w) circle (1.5pt);
    \node[blue,above right,inner sep=3pt] at (w) {$w$};

\end{tikzpicture}

\caption{
    A canonical shortest $s$--$t$ path $\rho$ in $G\setminus \{e,e'\}$ in Case~4. The blue curve represents $\rho$.
}
\label{fig:case4}
\end{figure}

By the choice of $u$, the prefix $\rho[s,u]$ contains no edge of $\pi[v_b, t]$.
This is because every vertex of $\rho[s,u)$ has distance from $s$ strictly less than $b$.
Similarly, by the choice of $w$, the suffix $\rho[w,t]$ contains no edge of $\pi[s, v_c]$.
Define
\begin{equation*}
A_e(u) := \dist_{(H+\pi[s,v_b])\setminus e}(s,u), \qquad B_{e'}(w) := \dist_{(H+\pi[v_c,t])\setminus e'}(w,t).
\end{equation*}

We next consider the middle portion $\rho[u,w]$.
Let $v_f$ be the first vertex of $\rho[u,w]$ in $\pi[v_{i+1},v_b]$, and let $v_{g}$ be the last vertex of $\rho[u,w]$ in $\pi[v_c,v_j]$.
Such vertices exist since $\rho[u,w]$ contains both $v_x$ and $v_y$.

To analyze the two subpaths $\rho[u,v_f]$ and $\rho[v_g,w]$, we first establish the following claim.

\begin{claim}\label{clm:case4}
    Every vertex $v_r\in V(\pi)\cap V(\rho[u,w])$ satisfies $i+1 \le r\le j.$
\end{claim}

\begin{proof}
Suppose first that $\rho[u,w]$ visits some $v_r$ with $r\le i$.
Since $\pi[s,v_r]$ contains no failed edge, canonicality implies $\rho[s,v_r]=\pi[s,v_r]$.
However, $\rho[s,v_r]$ contains $u$, while $\dist_G(s,u)=b>r$, a contradiction.

Similarly, if $\rho[u,w]$ visits some $v_r$ with $r\ge j+1$, then $\pi[v_r,t]$ contains no failed edge, and canonicality implies $\rho[v_r,t]=\pi[v_r,t]$. Since $\rho[v_r,t]$ contains $w$, this would imply $w\in V(\pi[v_r,t])$. However, every vertex on $\pi[v_r,t]$ has distance at least $r$ from $s$, whereas $\dist_G(s,w)=c<r$, a contradiction.
\end{proof}

By \cref{clm:case4} and the choice of $v_f$, the subpath $\rho[u,v_f]$ contains no edge of $\pi[s,v_b]$.
Indeed, otherwise it would use some edge $(v_r,v_{r+1})$ with $i+1 \le r<b$, and hence visit a vertex of $\pi[v_{i+1},v_b]$ before $v_f$, contradicting the choice of $v_f$.
Similarly, by \cref{clm:case4} and the choice of $v_g$, the subpath $\rho[v_g,w]$ contains no edge of $\pi[v_c,t]$.
Hence, $\rho[u,v_f]\subseteq H+\pi[v_b,t]$ and $\rho[v_g,w]\subseteq H+\pi[s,v_c]$.

Now, define
\begin{equation*}
U_{e'}^u(f)
:=
\dist_{(H+\pi[v_b,t])\setminus e'}(u,v_f),
\qquad
V_e^w(g)
:=
\dist_{(H+\pi[s,v_c])\setminus e}(v_{g},w).
\end{equation*}

There are two cases depending on the order of $v_f$ and $v_g$ on $\rho$.

\subparagraph*{Case 4.1:  $v_g$ appears at or before $v_f$ on $\rho$.}
In this case, $v_f$ lies on $\rho[v_g,w]$, and hence the middle portion $\rho[u,w]$ can be split at $v_f$ into $\rho[u,v_f]$ and $\rho[v_f,w]$.
In this case, we decompose $\rho[u, w]$ as $\rho[u, v_f] \circ \rho[v_f, w]$.
Accordingly, define
\begin{equation*}
M_{e,e'}^{u,w}
:=
\min_{i+1\le f\le b}
\left\{
U_{e'}^u(f)+V_e^w(f)
\right\}.
\end{equation*}

\subparagraph*{Case 4.2: $v_f$ appears before $v_g$ on $\rho$.}
Since $i+1\le f\le b\le c\le g\le j$
the subpath $\pi[v_f,v_g]$ contains neither failed edge.
By canonicality, $\rho[v_f,v_g]=\pi[v_f,v_g]$.
Thus, the two middle portions $\rho[u,v_f]$ and $\rho[v_g,w]$ are connected by the subpath $\pi[v_f,v_g]$.
In this case, we decompose $\rho[u, w]$ as $\rho[u, v_f] \circ \pi[v_f,v_g] \circ \rho[v_g, w]$.
Accordingly, define
\begin{align*}
C_{e,e'}^{u,w}
&:=
\min_{\substack{i+1\le f\le b\\ c\le g\le j}}
\left\{
U_{e'}^u(f)
+(g-f)
+V_e^w(g)
\right\}\\
&=
\min_{i+1\le f\le b}
\left\{
U_{e'}^u(f)-f
\right\}
+
\min_{c\le g\le j}
\left\{
g+V_e^w(g)
\right\}.
\end{align*}

\subparagraph*{Combining Cases~4.1 and~4.2.}
Accordingly, define
\begin{equation*}
D_{\mathrm B}(e,e') := \min_{u\in S_p,w\in T_q}
\left\{
A_e(u)
+
\min\left\{
M_{e,e'}^{u,w},
C_{e,e'}^{u,w}
\right\}
+
B_{e'}(w)
\right\}.
\end{equation*}
Every term in the minimum corresponds to an $s$--$t$ walk in $G\setminus \{e,e'\}$.
Conversely, the canonical path $\rho$ determines suitable vertices $u,w, f, g$, and hence one of the above terms has length $|\rho|$.
Therefore, in Case~4, we have $|\rho| = D_{\mathrm B}(e,e')$.

\subparagraph*{Computing $A$, $B$, $U$, and $V$.}
All required values $A$, $B$, $U$, and $V$ can all be computed using SSRP.
For each block $\gamma_p=\pi[v_a,v_b]$, we run SSRP from $s$ in $H+\pi[s,v_b]$ to compute all values $A_e(u)$.
Similarly, for each block $\gamma_q=\pi[v_c,v_d]$, we run SSRP from $t$ in $(H+\pi[v_c,t])^{\mathrm{rev}}$ to compute all values $B_{e'}(w)$.

For each vertex $u\in S_p$, we run SSRP from $u$ in $H+\pi[v_b,t]$ to compute all values $U_{e'}^u(z)$.
Similarly, for each vertex $w\in T_q$, we run SSRP from $w$ in $(H+\pi[s,v_c])^{\mathrm{rev}}$ to compute all values $V_e^w(z)$.

By Lemma~\ref{thm:msrp}, one call to SSRP takes $\Ot(m\sqrt{n}+n^2)$ time.
There are $O(n/L)$ blocks, and by the definition of the modulo-layer hitting set, we have $\sum_p |S_p|=O(n/L)$ and $\sum_q |T_q|=O(n/L)$. 
Hence, the total number of SSRP calls is $O(n/L)$.
Therefore, all required values $A$, $B$, $U$, and $V$ can be computed in $\tO\left((mn^{3/2}+n^3)/L\right)$ time.

\subparagraph*{Computing $D_{\mathrm B}$.}
Fix $u\in S_p$ and $w\in T_q$.
There are $O(L)$ possible choices for $e\in E(\gamma_p)$ and $O(L)$ possible choices for $e'\in E(\gamma_q)$.
For each fixed pair $(e,e')$, the value $M_{e,e'}^{u,w}$ can be computed by scanning the $O(L)$ possible indices $f$.
Similarly, each of $\min_{i+1\le f\le b} \left\{ U_{e'}^u(f)-f \right\}$ and $\min_{c\le g\le j} \left\{ g+V_e^w(g) \right\}$can be computed in $O(L)$ time.
Hence, $C_{e,e'}^{u,w}$ can also be computed in $O(L)$ time.
Therefore, all required values $M_{e,e'}^{u,w}$ and $C_{e,e'}^{u,w}$ for the fixed pair $(u,w)$ can be computed in $O(L^3)$ time.

The total number of pairs $(u,w)$ over all pairs of blocks is at most
\[
\sum_{p<q}|S_p||T_q|
\le
\left(\sum_p|S_p|\right)
\left(\sum_q|T_q|\right)
=
O\left(\frac{n^2}{L^2}\right).
\]
Consequently, the total time for computing all values $M_{e,e'}^{u,w}$ and $C_{e,e'}^{u,w}$ is
$
O\left(
\frac{n^2}{L^2}\cdot L^3
\right)
=
O(n^2L).
$

Finally, all values $D_{\mathrm B}(e,e')$ can be computed in
$
\Ot\left(
\frac{mn^{3/2}+n^3}{L}
+n^2L
\right)
$
time.

\paragraph*{Total running time.}
From Cases~1--4, all values $D_{\mathrm{L}}(e,e'), D_{\mathrm{R}}(e,e'),D_{\mathrm{F}}(e,e'),D_{\mathrm{B}}(e,e')$ can be computed in $ \Ot\left( \frac{mn^{3/2}+n^3}{L} +n^2L \right)$ time.
Therefore, all values $\dist_{G\setminus \{e,e'\}}(s,t)$ for pairs $(e,e')$ lying in different blocks can be computed in $\Ot\left((mn^{3/2}+n^3)/L+n^2L\right)=\tO(n^{3-1/4})$ total time, recalling that $L=\lceil n^{3/4}\rceil$.

\section{Concluding Remarks}

In this work, we introduced the \emph{modulo-layer hitting set} technique and showed that this simple idea is useful for both single-failure and dual-failure replacement paths in unweighted directed graphs. 
For RP, in particular, it leads to a remarkably simple deterministic algorithm that is suitable for textbook presentation and practical implementation.
We believe that this technique may be useful more broadly in other shortest-path problems.

For 2-FRP, the modulo-layer hitting-set technique allows us to solve the problem using MSRP and SSRP as subroutines.
Our algorithm can be viewed as a weak form of a reduction from directed 2-FRP to MSRP and SSRP. Roughly speaking, our algorithm reduces directed 2-FRP to a constant number of MSRP instances with $O(\sqrt n)$ sources and $O(n^{1/4})$ SSRP instances.

Finally, the only randomization in our 2-FRP algorithm comes from the MSRP algorithm by Harada and Izumi~\cite{HI26}.
Thus, a deterministic MSRP algorithm with the same running time as that of Harada and Izumi would immediately yield a deterministic $\Ot(n^{3-1/4})$-time algorithm for 2-FRP.

\paragraph{Acknowledgements.}
Kaito Harada was supported by JST BOOST, Japan Grant Number JPMJBS2402.

\paragraph{AI Disclosure.}
The main ideas behind our RP and 2-FRP algorithms were developed largely with the assistance of GPT-6 Astra (OpenAI).
The authors' contributions include simplifying several of the resulting algorithms, unifying the underlying techniques under the modulo-layer hitting set, developing the idea of using MSRP to improve the running time for 2-FRP when one failed edge lies on $\pi$, and the entire design and analysis of the algorithm for this case (\cref{thm:oneoff}).
We also used GPT-6 Astra (OpenAI) and Claude Fable 5.1 (Anthropic) to assist with language editing and improving the presentation of parts of the manuscript.
All proofs and the text of this paper were written and verified by the authors, who take full responsibility for its content.

\bibliographystyle{alphaurl}
\bibliography{bib2doi}

\clearpage
\appendix

\section{One Failure on $\pi$}\label{app:oneoff}
This appendix completes the proof of \cref{thm:oneoff}.
We describe the overall algorithm that computes all values $\dist_{G \setminus \{e, e'\}}(s, t)$ with $e \in E(\pi)$ and $e' \notin E(\pi)$ from the distances of \cref{alg:dist_compute_msrp}, and analyze its running time.

\paragraph*{Pairs to be computed.}
Recall from \cref{sec:overview-2frp} that $\rho_{e_i}$ denotes a fixed replacement path for $e_i = (v_i, v_{i+1}) \in E(\pi)$, and that only the second failures $e' \in E(\rho_{e_i}) \setminus E(\pi)$ need to be considered.
Let
\begin{equation*}
    \mathcal{K} := \bigl\{ (i, e') \bigm| 0 \le i < k,\ e' \in E(\rho_{e_i}) \setminus E(\pi) \bigr\},
\end{equation*}
where $E(\rho_{e_i}) := \emptyset$ if $e_i$ has no replacement path.
Then $|\mathcal{K}| < n^2$, and every pair of $\mathcal{K}$ has $e' \in E(H)$.

\paragraph*{Overall algorithm.}
Consider a pair $(i, e') \in \mathcal{K}$.
\cref{lem:2frp-detour,lem:2frp-decomp} are \cref{lem:rp-detour,lem:rp-decomp} with $H$ replaced by $H \setminus e'$, and the derivation of \cref{sec:overview-rp} applies verbatim to them.
Hence
\begin{equation*}
    \dist_{G \setminus \{e_i, e'\}}(s, t) = \min\{ C_i(e'), S_i(e') \},
\end{equation*}
where $C_i(e')$ and $S_i(e')$ are the candidates $C_i$ and $S_i$ of \cref{sec:overview-rp} with $\dist_H$ and $\dist_{H_{q(i)}}$ replaced by $\dist_{H \setminus e'}$ and $\dist_{H_{q(i)} \setminus e'}$.
Accordingly, let $A^{e'}(r)$, $B^{e'}(r)$, and $W^{e'}_{q,a}$ denote the tables that the scans of \cref{alg:rp} build with these distances; then
\begin{equation*}
    C_i(e') = \min_{r \in \mathcal{R}} \bigl\{ A_i^{e'}(r) + B_i^{e'}(r) \bigr\}
    \qquad\text{and}\qquad
    S_i(e') = \min_{h_{q(i)-1} \le a \le i} \bigl\{ a + W_{q(i),a}^{e'}(i) \bigr\}.
\end{equation*}
In particular, running lines~\ref{line:rp-pivot-begin}--\ref{line:rp-return} of \cref{alg:rp} with the distances of $H \setminus e'$ computes $\min\{C_i(e'), S_i(e')\}$ for all $i$.
Doing so for every $e' \in E(H)$ takes $O(n^2/L + nL)$ time per edge, hence $O(mn^2/L + mnL)$ time in total, which is too slow.
The algorithm below avoids this by reusing the scans of \cref{alg:rp} across different $e'$, and by evaluating the candidates only for the pairs of $\mathcal{K}$.

\paragraph*{Reusing scans.}
Every scan of \cref{alg:rp} reads only the distances from a single source $u$ of \cref{alg:dist_compute_msrp}: from $r$ in $H^{\mathrm{rev}}$ for $A^{e'}(r)$, from $r$ in $H$ for $B^{e'}(r)$, and from $v_a$ in $H_q$ for $W^{e'}_{q,a}$.
Let $T_u$ be the shortest-path tree of $u$ fixed by the MSRP call (edges of $H^{\mathrm{rev}}$ identified with those of $H$), and write $T^{\mathrm{in}}_r$, $T^{\mathrm{out}}_r$, and $T_{q,a}$ for the three kinds of sources.
If $e' \notin E(T_u)$, removing $e'$ does not change the distances from $u$ (\cref{subsec:msrp}), so the table of $u$ for $e'$ equals the table $A(r)$, $B(r)$, or $W_{q,a}$ that \cref{alg:rp} builds for $H$, which we call the \emph{default table} of $u$.
Hence, for each $e' \in E(H)$, the algorithm rebuilds only the tables of the sources $u$ with $e' \in E(T_u)$, and lets every other table refer to the default table without copying it.
Since $T_u$ has at most $n - 1$ edges, each source rebuilds its table for at most $n - 1$ edges $e'$.
Over all $e'$, the $|\mathcal{R}| \le n/L$ pivots thus rebuild $O(n)$ tables of $k < n$ entries each, in $O(n^3/L)$ time, and the at most $L + 1$ path vertices of a block $V_q$ rebuild at most $|V_q| - 1$ tables of at most $L$ entries each, in $O(nL^2)$ time since $\sum_q |V_q| \le 2n$.
Whether $e' = (x, y)$ lies in $E(T_u)$ is decided in $O(1)$ time by checking whether $x$ is the parent of $y$ in $T_u$ (or $y$ of $x$ for $T^{\mathrm{in}}_r$), and $e'$ can lie in $T_{q,a}$ only for the at most two blocks $V_q$ that contain both endpoints of $e'$.
The decisions thus take $O(m(n/L + L))$ time in total, which is dominated by the time of \cref{alg:dist_compute_msrp}.

\begin{algorithm}[!htb]
    \caption{The algorithm for one failure on $\pi$.}\label{alg:oneoff}
    \Input{An unweighted directed graph $G = (V, E)$ and vertices $s, t \in V$.}
    \Output{$\dist_{G \setminus \{e_i, e'\}}(s, t)$ for every $(i, e') \in \mathcal{K}$.}
    \BlankLine
    $L \gets \lceil \sqrt{n} \rceil$, and run \cref{alg:rp-blocks,alg:dist_compute_msrp}\;\label{line:oneoff-pre}
    Run lines~\ref{line:rp-pivot-begin}--\ref{line:rp-return} of \cref{alg:rp} with the distances of $H$, which gives the default tables $A(r)$, $B(r)$, $W_{q,a}$ and the candidates $C_i$, $S_i$\;\label{line:oneoff-default}
    Compute the paths $\rho_{e_i}$ for all $0 \le i < k$ and the set $\mathcal{K}$\;\label{line:oneoff-paths}
    \BlankLine
    \ForEach{$e' \in E(H)$}{\label{line:oneoff-loop-begin}
        \tcp{Rebuild the tables whose source tree contains $e'$}
        \ForEach{$r \in \mathcal{R}$}{\label{line:oneoff-pivot-begin}
            \lIf{$e' \in E(T^{\mathrm{in}}_r)$}{$A^{e'}(r) \gets$ lines~\ref{line:rp-scanA-begin}--\ref{line:rp-scanA-end} of \cref{alg:rp} with $\dist_{H \setminus e'}(\ast, r)$}
            \lElse{$A^{e'}(r) \gets A(r)$}
            \lIf{$e' \in E(T^{\mathrm{out}}_r)$}{$B^{e'}(r) \gets$ lines~\ref{line:rp-scanB-begin}--\ref{line:rp-scanB-end} of \cref{alg:rp} with $\dist_{H \setminus e'}(r, \ast)$}
            \lElse{$B^{e'}(r) \gets B(r)$}\label{line:oneoff-pivot-end}
        }
        \ForEach{block $V_q$ with $e' \in E(H_q)$ and each $v_a \in V(\gamma_q)$}{\label{line:oneoff-block-begin}
            \lIf{$e' \in E(T_{q,a})$}{$W^{e'}_{q,a} \gets$ lines~\ref{line:rp-scanW-begin}--\ref{line:rp-scanW-end} of \cref{alg:rp} with $\dist_{H_q \setminus e'}(v_a, \ast)$}
            \lElse{$W^{e'}_{q,a} \gets W_{q,a}$}\label{line:oneoff-block-end}
        }
        \BlankLine
        \tcp{Evaluate the candidates for the pairs of $\mathcal{K}$}
        \ForEach{$i$ with $(i, e') \in \mathcal{K}$}{\label{line:oneoff-eval-begin}
            $C_i(e') \gets \min_{r \in \mathcal{R}} \{A_i^{e'}(r) + B_i^{e'}(r)\}$\;
            \lIf{$e' \in E(H_{q(i)})$}{$S_i(e') \gets \min_{h_{q(i)-1} \le a \le i} \{a + W^{e'}_{q(i),a}(i)\}$}
            \lElse{$S_i(e') \gets S_i$}
            $\dist_{G \setminus \{e_i, e'\}}(s, t) \gets \min\{C_i(e'), S_i(e')\}$\;\label{line:oneoff-eval-end}
        }
    }
    \BlankLine
    \Return $\dist_{G \setminus \{e_i, e'\}}(s, t)$ for all $(i, e') \in \mathcal{K}$\;\label{line:oneoff-return}
\end{algorithm}

\paragraph*{Running time.}
Line~\ref{line:oneoff-pre} of \cref{alg:oneoff} takes $O(m)$ time for \cref{alg:rp-blocks} and, as computed in \cref{sec:overview-2frp}, $\Ot(mn/\sqrt{L} + n^3/L + m\sqrt{nL} + n^2 L)$ time for \cref{alg:dist_compute_msrp}, whose output we store so that every distance is accessed in $O(1)$ time.
Line~\ref{line:oneoff-default} takes $O(n^2/L + nL)$ time by \cref{sec:overview-rp}, and line~\ref{line:oneoff-paths} takes $O(n^2)$ time, since $\rho_{e_i}$ is read off in $O(n)$ time from a minimizer of $\min\{C_i, S_i\}$ recorded in line~\ref{line:oneoff-default} and the shortest-path trees of the MSRP calls, as the concatenation of $\pi[s, v_a]$, a shortest $v_a$--$v_b$ path of $H$ through a pivot or inside $H_{q(i)}$, and $\pi[v_b, t]$.
Over all $e' \in E(H)$, lines~\ref{line:oneoff-pivot-begin}--\ref{line:oneoff-block-end} take $O(m(n/L + L) + n^3/L + nL^2)$ time by the analysis above, and lines~\ref{line:oneoff-eval-begin}--\ref{line:oneoff-eval-end} take $O(n/L + L)$ time per pair of $\mathcal{K}$, hence $O(n^3/L + n^2 L)$ time.
Since $nL^2 \le n^2 L$, $mn/L \le mn/\sqrt{L}$, and $mL \le m\sqrt{nL}$, the total running time is
\begin{equation*}
    \Ot\left(\frac{mn}{\sqrt{L}} + \frac{n^3}{L} + m\sqrt{nL} + n^2 L\right).
\end{equation*}
For $L = \lceil \sqrt{n} \rceil$, this is $\Ot(mn^{3/4} + n^{5/2})$, which proves \cref{thm:oneoff}.

\end{document}